\documentclass[a4paper,12pt]{article}
\pdfoutput=1
\usepackage[utf8]{inputenc}
\usepackage{authblk}
\usepackage{cite,amsmath,amsfonts,amsthm,fullpage}
\usepackage{youngtab}
\usepackage{graphics}
\usepackage{miniplot}
\usepackage{subfigure}
\newcommand{\hl}{\emph}

\newcommand{\bpow}{\mathbf{p}}
\newcommand{\bbpow}{{\bar{\mathbf{p}}}}

\newcommand{\ttr}{\texttt{tr}}

\newcommand{\mc}{\textsc{m}}

\newcommand{\tig}{\tilde{\Gamma}}
\newcommand{\f}{\textsc{f}}
\newcommand{\e}{\textsc{e}}
\newcommand{\V}{\textsc{v}}
\newcommand{\E}{\mathbb{E}}
\newcommand{\Dr}{{\cal {L}}}
\theoremstyle{plain}
\newtheorem{Theorem}{Theorem}
\newtheorem{Lemma}{Lemma}
\newtheorem{Proposition}{Proposition}
\newtheorem{Corollary}{Corollary}
\newtheorem{Remark}{Remark}
\newtheorem{Example}{Example}
\theoremstyle{remark}
\theoremstyle{example}

\def\tr{\mathrm {tr}}
\def\det{\mathrm {det}}

\def\diag{\mathrm {diag}}

\def\G{\mathbb{GL}_N(\mathbb{C})}

\def\bp{\begin{Proposition}}
\def\ep{\end{Proposition}}
\def\bc{\begin{Corollary}}
\def\ec{\end{Corollary}}
\def\bl{\begin{Lemma}}
\def\el{\end{Lemma}}
\def\be{\begin{equation}}
\def\ee{\end{equation}}
\def\br{\begin{Remark}}
\def\er{\end{Remark}}
\def\brs{\begin{remarks}.\\ \rm\
\begin{enumerate}}
\def\ers{\end{enumerate}\end{remarks}}
\def\bea{\begin{eqnarray}}
\def\eea{\end{eqnarray}}
\def\bx{\begin{Example}}
\def\ex{\end{Example}}

\usepackage{environ}
\NewEnviron{myequation}{%
\begin{equation}
\scalebox{0.88}{$\BODY$}
\end{equation}}

\def\tr{\mathrm {tr}}
\def\det{\mathrm {det}}

\def\diag{\mathrm {diag}}

\def\&{&{\hskip -20pt}}

\newcount\YDcount\YDcount=0
\def\YDsize{10pt}
\def\YD#1{%
\ifnum#1=0
 \ifnum\YDcount=0 \ifx\varnothing\undefined\emptyset\else\varnothing\fi
 \else\vskip1.4pt\egroup\YDcount=0\fi
\else
 \ifnum\YDcount=0 \YDcount=1\vcenter\bgroup\vskip1pt
 \else\nointerlineskip\fi
 \vbox{\hrule\hbox{\vrule height\YDsize
 \loop\hskip\YDsize\vrule\ifnum\YDcount<#1\advance\YDcount1\repeat}\hrule
 \kern-0.4pt}\expandafter\YD
\fi}

\title{On Products of Random Matrices
}

\author[1,2,3,4]{Natalia Amburg}
 \affil[1]{%
A.I. Alikhanov  Institute for Theoretical and Experimental Physics of NRC “Kurchatov Institute”, B. Cheremushkinskaya, 25, Moscow, 117259, Russia
}
\affil[2]{%
National Research University “Higher School of Economics” , Myasnitskaya Ul., 20, Moscow, 101000, Russia
}%
\affil[3]{%
\quad Institute for Information Transmission Problems of RAS (Kharkevich Institute), Bolshoy Karetny per. 19, build.1, 127051  Moscow, Russia
} 
\affil[4]{%
\quad Moscow Center for Fundamental and Applied Mathematics, 119991 Moscow, Russia
} 
\author[1,5]{Aleksander Orlov} 
 \affil[5]{%
 Institute of Oceanology, Nahimovskii Prospekt 36, 117997 Moscow, Russia
}
\author[1,3,6]{Dmitry Vasiliev}
 \affil[6]{%
 \quad Moscow Institute of Physics and Technology, 141701 Dolgoprudny, Russia
}
 \date{August 30, 2020}                     
\begin{document}
\maketitle

\abstract{We introduce a family of models, which we  name matrix models associated with children's drawings---the so-called
dessin d'enfant. 
Dessins d'enfant are graphs of  a special kind drawn on a closed connected orientable surface (in the sky). 
The vertices of such a graph are small disks that we call stars. We attach random matrices to the edges of the graph
and get multimatrix~models. 
Additionally, to the stars we attach source matrices. They play the role of free parameters or model coupling~constants. 
The answers for our integrals are expressed through quantities that we call the "spectrum of stars." The answers 
may also include some combinatorial numbers, such as Hurwitz numbers or characters from group representation theory.}

\bigskip



\section{Introduction}

Interest in matrix integrals arose in different contexts and at different times. These are problems of statistics 
in biology (Wishart), and problems of quantum chaos (Wigner, Dyson, Gorkov, Eliashberg, Efetov), and purely 
mathematical problems of representation theory (see the \hl{textbook} 
 \cite{VKIII}). At the end 
of the previous century, applications were added in the theory of elementary particles (t'Hooft~\cite{t'Hooft}),
statistical physics (Migdal, Itsikson and Zuber \cite{Itzykson-Zuber}, and Kazakov and Kostov \cite{KazakovKostovNekrasov})
and string theory (Kazakov and Brezin~\cite{BrezinKazakov}, and Migdal and Gross \cite{GrossMigdal}). Now new applications 
have been added: in the theory of information transfer and theory of quantum information (\cite{Collins2016} or lectures \cite{Preskill} (Ch. 10 on quantum Shannon theory)). We 
allow ourselves not to provide an ocean of links for each of these areas.

Moving closer to the point, we only refer to works on the use of random matrices in \cite{Tulino},
works~about 
products of complex random matrices \cite{Ak1,Ak2,AkStrahov}, a review \cite{Ipsen,S1,S2,Chekhov-2014} and works close to ours from a mathematical point of view 
\cite{Kazakov,Kazakov2,Kazakov3,Kazakov-SolvMM,Alexandrov,ChekhovAmbjorn}.

In the theory of information, the product of the matrices describes the cascade transformation of signals, and 
averaging over the matrices means introducing interference and noise. The task is to calculate various correlation 
functions in such models.

We offer a family of models that in a sense can be called exactly solvable (cf. \cite{Kazakov-SolvMM}). 
They are built according to the 
so-called children's drawings, more precisely, clean children's drawings, which in combinatorics are also called maps.
This is a graph drawn on a closed orientable surface, which has the following property: if we cut it along the edges, 
the surface decomposes into regions homeomorphic to~disks (that is, they can be turned into disks by continuous
transformation). In addition to this picture, we turn the vertices of the graph into small disks; we call these 
disks stars.

The edges of the graph are assigned matrices over which the integration is performed. Depending on the~graph we draw, we will
get one model or another.

Surprisingly, it turns out that studying the products of random matrices with sources is an easier 
and more natural task. Writing answers for such integrals turns out to be a faster task if we have these~matrices. 
The absence of these matrices is equivalent to some additional averaging that needs to be specifically monitored.

Writing out some answers requires knowledge of certain combinatorial numbers for which tables~exist. In one case, 
these 
are the so-called Hurwitz numbers; in the other, these are the characters of representations of a symmetric group. 
However, sometimes the answers are simplified and written in \mbox{a fairly} simple form-in the form of determinants or in 
the form of products.

Adding source matrices leads to a variety of interesting relationships with differential operators. We  briefly 
mention this in the Appendix \ref{DiffOp}.

\section{Technical Tools}

\vspace{-6pt}

\subsection{Partitions; Power Sums and Schur Functions; Hurwitz Numbers
\label{Partitions-and-Schur-functions}}

Here we follow \cite{Mac}. Further technical details are in Appendix \ref{Partitions-and-Schur-functions-} and Appendix \ref{IOUG} 

{Partitions.} The \textit{partition} 
$\lambda=(\lambda_1,\dots,\lambda_k)$ is a set of nonnegative integers $\lambda_i$ which are called
parts of $\lambda$ and which are ordered as $\lambda_i \ge \lambda_{i+1}$. 
The number of non-vanishing parts of $\lambda$ is called the length of the partition $\lambda$, and will be denoted by
 $\ell(\lambda)$. The number $|\lambda|=\sum_i \lambda_i$ is called the weight
 of $\lambda$.  The set of all partitions will be denoted by $\Upsilon$ and the set of partitions of weight $d$ 
 by $\Upsilon_d$. Example: the~partition~$(4,4,1)$ belongs to $\Upsilon_9$. The length of $(4,4,1)$ is 3.

 We shall use Greek letters for partitions.
 
Sometimes, it is convenient to use a different notation
 \[
  \lambda =\left(1^{m_1}2^{m_2}3^{m_3}\cdots \right),
 \]
where $m_k$ is the number of times an integer $k$ occurs as a part of $\lambda$. For instance, $(4,4,1)$ may be written 
as $\left(1^14^2 \right)$. The number
\be\label{z-lambda}
z_\lambda =\prod_{k>0} k^{m_k}m_k!
\ee
plays an important role hereinafter. 

Example: $z_{(4,4,1)}=z_{(1^14^2)}=1^11!4^22!=32$.
 
 Partitions can be perceived visually using Young diagrams (YDs): a partition $\lambda$ with parts 
 $\lambda_1,\dots,\lambda_\ell$ matches 
 the rows of the corresponding YD of the length $\lambda_1,\dots,\lambda_\ell$, respectively; see \cite{Mac}
 for details. The weight of $\lambda$ is the area of the YD of $\lambda$.

 {Power sums.}
For a matrix $X\in GL_N(\mathbb{C})$, we put
 \be\label{Newton-sums}
p_k(X) := \ttr X^k= x_1^k +\cdots + x_N^k,\qquad k=1,2,\dots,
\ee
which are the Newton sums of its eigenvalues.

Then, for a partition $\Delta=(\Delta_1,\dots,\Delta_k)$, we introduce
  \be\label{power-sums}
 \bpow_\Delta(X) = \ttr( X^{\Delta_1})\dots 
\ttr( X^{\Delta_\ell}).
 \ee
 
 We put $\bpow_{(0)}(X)=1$.

  \br\label{X-to-C}
 Let us note that $\bpow_\Delta(X)$ is a polynomial in the eigenvalues of $X$ and \textit{also} a polynomial in
 entries of $X$. We consider $\bpow_\Delta(X)$ to be a map $GL_N(\mathbb{C})\rightarrow \mathbb{C}$.
 \er
The polynomial $\bpow_\Delta$ is a symmetric function of the eigenvalues $x_1,\dots,x_N$ and called
the \textit{power sum} labeled with a multi-index $\Delta$. 

If one assigns a degree 1 to each $x_i$, then the degree of $\bpow_\Delta(X)$ is $|\Delta|$.

In many problems, the set $p_1,p_2,\dots$ is considered as an independent set of complex parameters, which are in 
no way associated with any matrix. In this case, instead of $\bpow(X)$ we write simply $\bpow=(p_1,p_2,\dots)$.
The degree of $p_m$ is equal to $m$.


{Schur functions.} 
Next, let us recall the definition of the \textit{Schur polynomial}, parametrised by \mbox{a partition.} 
Consider~the~equality
\[
 e^{\sum_{m>0} \frac 1m z^m  p_m } = 1+z p_1+z^2\frac{p_1^2+p_2}{2}+\dots =\sum_{m\ge 0} z^m s_{(m)}(\bpow).
\]

The polynomials $s_{(m)}$ are called elementary Schur functions. 
As we can see, $s_{(0)}(\bpow)=1$.
Let $\lambda=[\lambda_1,\lambda_2,\dots]$ 
be a partition, and define the polynomial $s_\lambda$ by 
\be\label{Schur-det-p}
 s_\lambda(\bpow) = \det\left( s_{(\lambda_i-i+j)}(\bpow)  \right)_{i,j\ge 1}.
\ee
 (On the right-hand side, it is assumed that $s_{(m)}=0$ for negative $m$.)
 Here we define the Schur function as a function of the set of free parameters $\bpow=(p_1,p_2,\dots)$.
 
 Let us remember that we chose the variables $\bpow=(p_1,p_2,\dots)$ to be related to a matrix 
 $X\in \mathbb{GL}_N(\mathbb{C})$.
 Then $s_\lambda$ is a map $\mathbb{GL}_N(\mathbb{C}) \to \mathbb{C}$; we will write 
 $$
 s_\lambda(X):=s_\lambda(\bpow(X)) .
 $$
\br\label{s=polynom} 
Let us note that $s_\lambda(X)$ is a polynomial in entries of $X$ and a symmetric polynomial in the 
eigenvalues $x_1,\dots,x_N$ of degree $d=|\lambda|$.
 \er
 
\begin{Remark}[\cite{Mac}]  The Schur functions $s_\lambda(x_1,\dots, x_N)$, where $\ell(\lambda)\ge N$, 
 form a $\mathbb{Z}$-basis of the space of symmetric 
 polynomials in $x_1,\dots, x_N$ of degree $d$.
\end{Remark}
In terms of the eigenvalues $x_1,\dots,x_N$, the Schur function reads as
 \be\label{Schur-x}
 s_\lambda(X)=\frac{\det \left[x_j^{\lambda_i-i+N}\right]_{i,j\le N}}{\det \left[x_j^{-i+N}\right]_{i,j\le N}}
 \ee
 for $\ell(\lambda)\le N$, and vanishes for $\ell(\lambda) > N$. One can see that $s_\lambda(x)$ is a 
 symmetric homogeneous polynomial of degree $|\lambda|$ in the variables $x_1,\dots,x_N$.
 \br The polynomial
 $s_\lambda$ is the character of the irreducible representation of the $\mathbb{GL}_N(\mathbb{C})$ group
 labeled by~$\lambda$.
 \er

{Character map relation.}  
 Relation (\ref{Schur-det-p}) relates polynomials $s_\lambda$
 and $\bpow_\Delta$ of the same degree $d=|\lambda|=|\Delta|$.
 Explicitly, one can~write
\be\label{powersum-Schur}
\bpow_{\Delta}=\sum_{\lambda\in\Upsilon_d} 
\frac{{\rm dim}\lambda}{d!}
z_\Delta\varphi_\lambda(\Delta) s_\lambda(\bpow)
\ee
 and
\be\label{Schur-char-map}
s_\lambda(\bpow)=\frac{{\rm dim}\lambda}{d!}\sum_{\Delta\in\Upsilon_d } 
\varphi_\lambda(\Delta)\bpow_{\Delta}.
\ee

The last relation is called the character map relation.
Here 
\be\label{dim}
 \frac{{\rm dim}\lambda}{d!}: =
\frac{\prod_{i<j\le N}^{}(\lambda_i-\lambda_j-i+j)  }{\prod_{i=1}^{N}(\lambda_i-i+N)!}
=s_\lambda(\bpow_\infty)
\ee
(see example 1 in Section 1 and example 5 in Section 3 of Chapter I in \cite{Mac}), where 
\be\label{p-infty}
\bpow_\infty = (1,0,0,\dots).
\ee
and where   $N\ge \ell(\lambda)$. As one can check, the right-hand side does not depend on $N$.
(We recall that $\lambda_i=0$ in case $i>\ell(\lambda)$). The number
${\rm dim}\lambda$ is an integer.

The factors $\varphi_\lambda(\Delta)$ satisfy the following orthogonality relations.
\be\label{orth1}
\sum_{\lambda \in\Upsilon_d} \left(\frac{{\rm\dim}\lambda}{d!}\right)^2\varphi_\lambda(\mu)\varphi_\lambda(\Delta) =
 \frac{ \delta_{\Delta,\mu} }{z_{\Delta}}
\ee
and
\be\label{orth2}
\left(\frac{{\rm\dim}\lambda}{d!}\right)^2
\sum_{\Delta\in\Upsilon_d} z_\Delta\varphi_\lambda(\Delta)\varphi_\mu(\Delta) =
\delta_{\lambda,\mu}.
\ee

\br\label{char-char} Equality (\ref{Schur-char-map}) expresses the $\mathbb{GL}_N$ character $s_\lambda$ in terms of
characters $\chi_\lambda$ of the symmetric group $S_d$ labeled by the same partition $\lambda$
and evaluated on cycle classes $\Delta\in\Upsilon_d$. 
This explains the name of (\ref{Schur-char-map}).
\mbox{The numbers} $\varphi_\lambda(\Delta)$ are
called the normalized characters:
\[
 \chi_\lambda(\Delta)= \frac{{\rm\dim}\lambda}{d!} z_\Delta  \varphi_\lambda(\Delta).
\]

The integer ${\rm\dim}\lambda$ is the dimension of the representation $\lambda$; that is,
${\rm\dim}\lambda=\chi_\lambda\left((1^d)  \right)$. Equations~(\ref{orth1})~and~(\ref{orth2})
are the orthogonality relation for the characters.

\er

{Hypergeometric tau functions and determinantal formulas.} 

Here we follow \cite{OS-2000,OS-TMP}.
Let $r$ be a function on the lattice $\mathbb{Z}$.
Consider the following series
\be\label{tau-elementary}
 1+r(n+1)x+r(n+1)r(n+2)x^2+r(n+1)r(n+2)r(n+3)x^3 + \cdots =: \tau_r(n,x)
\ee

Here $n$ is an arbitrary integer.
This is just a Taylor series for a function $\tau_r$ with a given $n$  written in a form that is as close as 
possible to typical 
hypergeometric series. Here $n$ is just a parameter that will be of use later. 
If as $r$ we take a rational (or trigonometric) function, we get a generalized (or basic) hypergeometric series.
Indeed, take 
\be\label{r-rational}
r(n)=\frac{\prod_{i=1}^p(a_i+n)}{n\prod_{i=1}^q(b_i+n)}
\ee

We obtain
$$
\tau_r(n,x)=\sum_{m\ge 0} \frac{\prod_{i=1}^p(a_i+n)_m}{m!\prod_{i=1}^q(b_i+n)_m} x^m=
{_{p}F}_q\left({a_1+n,\dots,a_p+n\atop b_1+n,\dots, b_q+n} \mid x\right)
$$
where 
\be
(a)_n=a(a+1)\cdots(a+n-1)=\frac{\Gamma(a+n)}{\Gamma(a)}
\ee
is the Pochhammer symbol.

One can prove the following formula 
which express a certain series over partitions in terms of~(\ref{tau-elementary})
(see, for instance, \cite{OS-2000}):
\be\label{tau(nXY)}
\tau_r(n,X,Y):= \sum_{\lambda} r_\lambda(n) s_\lambda(X)s_\lambda(Y) =
\frac{c_n}{c_{n-N}}\frac{\det \left[ \tau_r(n-N+1,x_iy_j) \right]_{i,j\le N}}
{\det \left[ x_i^{N-k} \right]_{i,k\le N}\det \left[ y_i^{N-k} \right]_{i.k\le N}}
\ee
 where $c_k=\prod_{i=0}^{k-1}\left(r(i)  \right)^{i-k}$  and where
\be\label{content-product}
r_\lambda(n):=\prod_{(i,j)\in\lambda} r(n+j-i),\quad r_{(0)}(n)=1
\ee
where the product ranging over all nodes of the Young diagram $\lambda$
is the so-called content product (which has the meaning of the generalized
Pochhammer symbol related to $\lambda$).
Let us test the formula for the simplest case $r\equiv 1$:
\be
\det\left[1-X\otimes Y\right]=\sum_{\lambda} s_\lambda(X)s_\lambda(Y)
=\frac{\det\left[( 1-x_iy_j )^{-1} \right]_{i,j\le N}}
{\det \left[ x_i^{N-k} \right]_{i,k\le N}\det \left[ y_i^{N-k} \right]_{i.k\le N}}
\ee

Sometimes we also use infinite sets of power sums $\bpow^{i}=(p^{(i)}_1,p^{(i)}_1,\dots )$ and instead of matrices
$X$ and $Y$ set:
\be\label{tau-p1-p2}
\tau_r(n,\bpow^1,\bpow^2):= \sum_{\lambda} r_\lambda(n) s_\lambda(\bpow^1)s_\lambda(\bpow^2)
\ee

An example $r\equiv 1$:
\be\label{C-Littlewood}
\tau_1(n,\bpow^1,\bpow^2):=e^{\sum_{m>0} \frac 1m p_m^{(1)} p_m^{(2)} }=\sum_{\lambda} s_\lambda(\bpow^1)s_\lambda(\bpow^2)
\ee

In case the function $r$ has zeroes, there exists a determinantal representation. 
Suppose $r(0)$; then~$r_\lambda(n)=0$ if $\ell(\lambda)>n$. Then
 \be\label{tau(npp)}
\tau_r(n,\bpow^1,\bpow^2)= \sum_{\lambda\atop\ell(\lambda)\le n} r_\lambda(n) s_\lambda(\bpow^1)s_\lambda(\bpow^2)
= c_n \det\left[ \partial_{p_1^{(1)}}^{a}\partial_{p_1^{(2)}}^{b} \tau_r(1,\bpow^1,\bpow^2) \right]_{a,b=0,\dots,n-1},
\ee
where $c_k=\prod_{i=1}^{k-1}\left(r(i)  \right)^{i-k}$, and where
$$
\tau_r(1,\bpow^1,\bpow^2)=1+\sum_{m>0}r(1)\cdots r(m)s_{(m)}(\bpow^1)s_{(m)}(\bpow^2);
$$
 see \cite{O-2004-New}. 
 
 In addition, there is 
the following formula:
\be\label{tau(nXp)}
\tau_r(n,X,\bpow)= \sum_{\lambda} r_\lambda(n) s_\lambda(X)s_\lambda(\bpow)=
\frac{\det\left[ x_i^{N-k}
\tau_r\left(n-k+1, x_i,\bpow\right) \right]_{i,k\le N}}
{\det \left[ x_i^{N-k} \right]}
\ee
where
\be
\tau_r\left(n, x_i,\bpow\right)=1+\sum_{m > 0}r(n)\cdots r(n+m) x_i^m s_{(m)}(\bpow)
\ee

Let us test it for $r\equiv 1$, $\bpow=\bpow_\infty:=(1,0,0,\dots)$:
\be
e^{\tr X}= \sum_\lambda s_\lambda(X)s_\lambda(\bpow_\infty)=
\frac{\det\left[ x_i^{N-k}
e^{x_i} \right]_{i,k\le N}}
{\det \left[ x_i^{N-k} \right]}
\ee

There are similar series
$$
\sum_{\lambda} r_\lambda(n) s_\lambda(X)
$$
which can be written as a Pfaffian \cite{OST-I}; however, we will not use them in the present text.

Some properties of the Schur functions. Let us consider
\bl \label{2-properties}
For $X\in\mathbb{GL}_N$, where $\det X\neq 0$, and for 
$\lambda=(\lambda_1,\dots,\lambda_N)$, $\ell(\lambda)\le N$, we have
\be
s_\lambda(X)\det (X^\alpha)=s_{\lambda + \alpha}(X)
\ee
where $\alpha$ is a nonnegative integer and
where $\lambda+\alpha$ denotes the partition with parts $(\lambda_1+\alpha,\lambda_2+\alpha,\dots,\lambda_N+\alpha)$,
in particular
\be
s_\lambda(\mathbb{I}_N)=s_{\lambda+\alpha}(\mathbb{I}_N).
\ee

Additionally
\be
\frac{ s_\lambda(\bpow_\infty)}{s_{\lambda+\alpha}(\bpow_\infty)}=
 \prod_{i=1}^N \frac{\Gamma(h_i+1+\alpha)}{\Gamma(h_i+1)}=
 \frac{(N+\alpha)_\lambda}{(N)_\lambda}
 \ee 
\el

{Hurwitz number.} Let $\e\le 2$ be an integer and 
$\Delta^1=\left(\Delta^1_1,\Delta^1_2,\dots   \right),\dots,
\Delta^k=\left(\Delta^k_1,\Delta^k_2,\dots   \right) \in\Upsilon_d$. One can show that 
\be\label{Mednyh}
H_\e\left(\Delta^1,\dots,\Delta^k \right)=
\sum_{\lambda\in\Upsilon_d}\left(\frac{{\rm dim} \lambda}{d!} \right)^\e
\varphi_\lambda(\Delta^1)\cdots\varphi_\lambda(\Delta^k)
\ee
is a rational number. This number is called the Hurwitz number, which is a popular combinatorial
object in many fields of mathematics (see, for instance, \cite{ZL,Okounkov-Pand-2006}) and also 
in physics \cite{Dijkgraaf}. The explanations of the geometrical and combinatorial meanings of Hurwitz numbers 
may be found 
in Appendices~\ref{Hurwitz-geometric-section}~and~\ref{Hurwitz-combinatorial-section}.
We also need a weighted Hurwitz number
\be\label{Mednyh-m}
H_\e\left(\Delta^1,\dots,\Delta^k |m\right)=
\sum_{\lambda\in\Upsilon_d}\left(\frac{{\rm dim} \lambda}{d!} \right)^\e
\varphi_\lambda(\Delta^1)\cdots\varphi_\lambda(\Delta^k)\left(\frac{1}{(N)_\lambda}\right)^m
\ee
 
\subsection{Mixed Ensembles of Random Matrices \label{MixedEnsembles}}
 A complex number
$\E_{n_1,n_2}\{f\}$ is the notation of the expectation values of a function $f$ which depends
on the entries of matrices $Z_1,\dots,Z_{n_1}\in \mathbb{GL}_N(\mathbb{C})$ and of matrices
$U_1,\dots,U_{n_2}\in \mathbb{U}_N$:
\be\label{mixed}
\E_{n_1,n_2}\{f\}=\int f(Z_1,\dots,Z_{n_1},U_1,\dots,U_{n_2})
d\Omega_{n_1,n_2},
\ee

\be
d\Omega_{n_1,n_2} =  \prod_{i=1}^{n_1}d\mu(Z_i) \prod_{i=1}^{n_2} d_*U_i ,
\ee
where $d_*U_i$  ($i=1,\dots,n_2$) is the Haar measure on $\mathbb{U}_N$ and where
$$
d\mu(Z_i)=c\prod_{a,b}e^{-\hbar^{-1}|(Z_i)_{a,b}|^2}d^2 (Z_i)_{a,b}
$$ 
is the Gaussian measure.
Here $\hbar$ is a parameter; usually it is chosen to be $N^{-1}$.

Each set of $Z_i$ and $d\mu(Z_i)$ is called complex Ginibre ensemble and the whole set
$Z_1,\dots,Z_{n_1}$ and $d\mu(Z_1),\dots,d\mu(Z_{n_1})$ are called $n_1$ \textit{independent complex Ginibre ensembles}.
The set $U_1,\dots,U_{n_2}$ and the measure $d_*U_1,\dots,d_*U_{n_2}$ are called $n_2$ 
\textit{independent circular ensembles}.
We assume each $\int d_*U _i=\int d\mu(Z_i)=1$. 

The ensemble of the matrices $Z_1,\dots,Z_{n_1},U_1,\dots,U_{n_2}$ together with the probability measure
$d\Omega_{n_1,n_2}$
(that is, with expectation values defined  by (\ref{mixed})) we call a $(n_1,n_2)$ \textit{mixed ensemble}.
We will consider mixed ensembles with $n=n_1+n_2$ random matrices.

\subsection{Integrals of Schur Functions and Integrals of Power Sums}

In what follows we  study expectation values of Schur functions and of power sums.
We need four key lemmas.

These lemmas should be known in parts corresponding to Schur functions, but at the moment I have not found all 
the lemmas along with the proofs, so I will try to fill this gap.

 In the Lemmas \ref{Lemma-s-Z}--\ref{Lemma-s-s-U} below $A$ and $B$  are $N\times N$ \textit{complex} matrices;
this point is the key.
 
 Everywhere in this section, $ d = 1, \dots, N $.

\bl 
\label{Lemma-s-Z}

(I) For any $\lambda\in\Upsilon_d$ we have
\be\label{s(ZAZB)} 
\E_{1,0}\left\{ s_\lambda(ZAZ^\dag B)\right\} = \hbar^{d}\frac{s_\lambda(A) s_\lambda(B)}{s_\lambda(\bpow_\infty)},
\ee
where $s_\lambda(\bpow_\infty)$ is given in (\ref{dim}).

(II)  Equivalently,  for any $\Delta\in \Upsilon_d$ we have
\be\label{Hpp}
\E_{1,0}\left\{\,
\bpow_\Delta(ZAZ^\dag B)\,\right\}\,=
\, z_\Delta \hbar^{d}\,\sum_{\Delta^a,\Delta^b\in\Upsilon_d}\,
H_2(\Delta,\Delta^a,\Delta^b)\,\bpow_{\Delta^a}(A)\,\bpow_{\Delta^b}(B),
\ee
where
\be\label{Hurwtiz-H(*,*,*)}
H_2(\Delta,\Delta^a,\Delta^b)=\sum_{\lambda\in\Upsilon_d}
\left(\frac{{\rm dim} \lambda}{d!} \right)^2
\varphi_\lambda(\Delta)\varphi_\lambda(\Delta^a)\varphi_\lambda(\Delta^b)
\ee
is the three-point Hurwitz number (\ref{Mednyh}).
\el

Note that Hurwitz numbers do not depend on the order of its arguments, in particular, for (\ref{Hurwtiz-H(*,*,*)}),
$H_2(\Delta,\Delta^a,\Delta^b)=H_2(\Delta,\Delta^b,\Delta^a)=\cdots=H_2(\Delta^b,\Delta^a,\Delta)$.
\begin{proof}

(i) Relation (\ref{s(ZAZB)}) is known for Hermitian $A,B$; see \cite{Mac}, chap. VII, section 5, example 5.
Taking into account Remark \ref{s=polynom}, we see that both sides of (\ref{s(ZAZB)}) can be analytically
continued as functions of matrix entries of $A$ and $B$. Therefore, (\ref{s(ZAZB)}) is correct for 
$A,B\in\mathbb{GL}_N(\mathbb{C})$.

(ii) Relation (\ref{Hpp}) follows from (\ref{s(ZAZB)}) and vice versa. To see this we use (\ref{orth1}).
For example, let us get (\ref{Hpp}) from (\ref{s(ZAZB)}). We replace all Schur functions in 
(\ref{s(ZAZB)}) with power sums in accordance with   
(\ref{Schur-char-map}). Then we multiply the both sides of relation (\ref{s(ZAZB)}) by 
$\varphi_\lambda(\Delta) {\rm \dim}\lambda $, 
	then we sum over $\lambda$. Then the first orthogonality relation (\ref{orth1}) and (\ref{Mednyh}) results in (\ref{Hpp}).

(iii) The alternative way to prove the Lemma is to start with the left-hand side of (\ref{Hpp}), where~$A,B\in\mathbb{GL}_N(\mathbb{C})$. 
Such integrals were considered in \cite{NO2020} in the context of generating of Hurwitz numbers
(see~Appendix \ref{Hurwitz-geometric-section}).
Using the results of \cite{NO2020} we state that it is equal to the 
right-hand side of~(\ref{Hpp}),
where~$H_2(\Delta,\Delta^a,\Delta^b)$ is the three-point Hurwitz number given by (\ref{Mednyh}),
namely, to right-hand side of~(\ref{Hurwtiz-H(*,*,*)}).
Then with the help of orthogonality relation
(\ref{orth1}), we derive (\ref{s(ZAZB)}), where $A,B\in\mathbb{GL}_N(\mathbb{C})$. 

\end{proof}

\br
Regarding the last point (iii): The volume of the article does not allow describing the construction of work 
\cite{NO2020,NO2020F}.
In short:
The derivation of the formula (\ref{Hurwtiz-H(*,*,*)}), and the derivation of more general 
formulas~(\ref{E-power1})~below, are based on the 
application of Wick's 
theorem and the realization of the fact that each Wick pairing corresponds to a certain gluing of surfaces 
from polygons: the surface obtained in the first order of the perturbation theory ($|\Delta|=d=1$) is basic.
(For instance, the torus can be obtained from a rectangular by gluing the opposite sides;
see (a) and (b) in \hl{Figure} \ref{fig:figure5}.) The surfaces  obtained in the following orders ($d>1$) are the surfaces which cover the basic one.
In this case, the Hurwitz numbers simply give a weighted number of possible covering surfaces, where the Young 
diagrams correspond to the so-called branching profiles. The basic surface in case (\ref{Hurwtiz-H(*,*,*)})
is a sphere, glued from the 2-gon; see \hl{Figure} \ref{fig:figure2}b below where $X_1=Z,\,X_{-1}=Z^\dag, \,C_1=A,\,C_{-1}=B$.
\er
Examples of (\ref{Hurwtiz-H(*,*,*)}). Suppose $\Delta=(2)$ on the left-hand side of (\ref{Hurwtiz-H(*,*,*)}).
It is known that $H_2((2),\Delta^a,\Delta^b)=1/2$ in two cases---where $\Delta^a=(2),\,\Delta^b=(1,1)$ and
where $\Delta^a=(1,1),\,\Delta^b=(2)$---and it is zero otherwise. Additionally $z_{(2)}=2$, see (\ref{z-lambda}). Thus,
$$
\E_{1,0}\left\{\,
\tr\left((ZAZ^\dag B)^2\right)\,\right\}\,=
\, \hbar^2\,\tr\left( A^2 \right)\,\left(\tr B \right)^2 
+
\, \hbar^2\,\left( \tr A \right)^2\,\tr\left( B^2 \right). 
$$

Next, suppose $\Delta=(1,1)$. As we know, $H_2((1,1),(2),(2))=1/2$ and vanishes for different choices
of $\Delta^a,\Delta^b$, and $z_{(1,1)}=2$; see (\ref{z-lambda}); therefore,
$$
\E_{1,0}\left\{\,
\left( \tr ZAZ^\dag B \right)^2\,\right\}\,=
\, \hbar^2\,\tr\left( A^2 \right)\,\tr \left( B^2 \right).  
$$

\begin{Corollary} Suppose $\det A \neq 0$, $\det B \neq 0$, and 
$\lambda=(\lambda_1,\dots,\lambda_N)$, $\ell(\lambda)\le N$. We get
 \be\label{detZZ} 
\E_{1,0}\left\{ s_\lambda(ZAZ^\dag B)\det\left((Z^\dag Z)^{p_0} \right) \right\} = 
\hbar^{d+p_0 N}\frac{s_\lambda(A) s_\lambda(B)}{s_\lambda(\bpow_\infty)} 
\frac{(N+p_0)_\lambda}{(N)_\lambda}
\ee

 In particular,
  \be\label{detZZ'} 
\E_{1,0}\left\{ s_\lambda(ZAZ^\dag)\det\left((Z^\dag Z)^{p_0} \right) \right\} = 
\hbar^{d+p_0N}s_\lambda(A)
(N+p_0)_\lambda
\ee
 
\end{Corollary}

\begin{proof}
In case $p_0$ is a natural number the proof is as follows:
 $$
 s_\lambda(ZAZ^\dag B)\det\left(  (ZAZ^\dag B )^{p_0}\right)=
 s_\lambda(ZAZ^\dag B)\det\left(  (ZZ^\dag  )^{p_0}\right)\det\left(  (A)^{p_0}\right)\det\left(  ( B )^{p_0}\right)
 = s_{\lambda+p_0} (ZAZ^\dag B)
 $$
 and from Lemma \ref{2-properties}:
 where $\lambda+p_0=(\lambda_1+p_0,\lambda_2+p_0,\dots,\lambda_N +p_0)$, $\ell(\lambda')=N$. We have
 \be
\frac{ s_\lambda(\bpow_\infty)}{s_{\lambda'}(\bpow_\infty)}=
 \prod_{i=1}^N \frac{\Gamma(h_i+1+p_0)}{\Gamma(h_i+1)}=
 \frac{(N+p_0)_\lambda}{(N)_\lambda}
 \ee
 
 As we see he right-hand side can by analytically continued as the function of $p_0$.
\end{proof}

\bl \label{Lemma-s-s-Z} 

(I) Suppose $\lambda\in\Upsilon_d$, and let $\nu$ be any partition. We have
\be\label{s(ZA)s(ZB)} 
\E_{1,0}\left\{ s_\lambda(ZA) s_\nu(Z^\dag B)\right\} = 
\hbar^{d}\delta_{\lambda\nu}\frac{s_\lambda(AB)}{s_\lambda(\bpow_\infty)}.
\ee
where $s_\lambda(\bpow_\infty)$ is given in (\ref{dim}).

(II) Equivalently, let $\Delta^a\in\Upsilon_d$ and $\Delta^b\in\Upsilon_{d'}$. Then
\be\label{Hp2}
\E_{1,0}\left\{\,
\bpow_{\Delta^a}(ZA)\, \bpow_{\Delta^b}(Z^\dag B)\,\right\}\,=\, z_{\Delta^a} z_{\Delta^b}\, 
\delta_{d,d'}
\hbar^{d}\,
\sum_{\Delta\in\Upsilon_d}\,
H_{2}(\Delta^a,\Delta^b,\Delta)\,\bpow_{\Delta}(AB),
\ee 
where $ H_{\mathbb{CP}^1}(\Delta,\Delta^a,\Delta^b)$  
is given by (\ref{Hurwtiz-H(*,*,*)}).
\el
	The identity (\ref{s(ZA)s(ZB)}) is well-known; for instance, see \cite{Mac}.
Relation (\ref{Hp2}) was proven independently in~\cite{NO2020}. To obtain (\ref{s(ZA)s(ZB)}) from (\ref{Hp2})  we replace
power sums under the integral with the left-hand side of (\ref{powersum-Schur}) and use (\ref{Mednyh}) and the orthogonality relations
(\ref{orth1}) and (\ref{orth2}).




Denote $N\times N$ identity matrix  $\mathbb{I}_N$. The following equality is known (see \cite{Mac}: combine
examples 4 and 5 of  Section 3 and example 1 from Section I of chapter I):
\be
s_\lambda(\mathbb{I}_N)=(N)_\lambda s_\lambda(\bpow_\infty).
\ee

Here
\be\label{N-lambda}
 (N)_\lambda :=(N)_{\lambda_1}(N-1)_{\lambda_2}\cdots (N-\ell+1)_{\lambda_\ell},
\ee
where $\lambda=(\lambda_1,\dots,\lambda_\ell)$, $\lambda_\ell \neq 0$ and
where $(a)_k:=a(a+1)\cdots (a+k-1)=\frac{\Gamma(a+k)}{\Gamma(a)}$ is the \mbox{Pochhammer symbol}.

\bl \label{Lemma-s-U} For any $\lambda\in\Upsilon_d$, we have
\begin{equation}\label{sAUBU^+1}
\E_{0,1}\left\{ s_\lambda(UAU^{-1}B)\right\}=
\frac{s_\lambda(A)s_\lambda(B)}{s_\lambda(\mathbb{I}_N)} \ .
\end{equation}
\el
The proof is contained in \cite{Mac} chap VII, Section 5, example 3.
\bl\label{Lemma-s-s-U} Suppose $\lambda\in\Upsilon_d$. For any $\mu$, we get
\begin{equation}\label{sAUU^+B'}
\E_{0,1}\left\{ s_\mu(UA)s_\lambda(U^{-1}B)\right\} =
\frac{s_\lambda(AB)}{s_\lambda(\mathbb{I}_N)}\delta_{\mu,\lambda}\,.
\end{equation}
\el
The proof follows from relations (1) and (3) chap VII, Section 5, example 1 in \cite{Mac}.

\br
Formula (\ref{sAUBU^+1}) gives the fastest derivation of the famous formula HCIZ:
$$
\int e^{\alpha\tr UAU^\dag B}d_*U=c\frac{\det \left[ e^{a_ib_j}\right]_{i,j}}{\prod_{i>j}(a_i-a_j)(b_i-b_j)}
$$

Indeed, we use (\ref{C-Littlewood}) where $\bpow^1=(1,0,0,\dots)=:\bpow_\infty$ and 
$\bpow^2=(p^{(2)}_1,p^{(2)}_2,\dots)$ with $p^{(2)}_m=\tr\left( (UAU^\dag B)^m \right)$;~thus,
$$
e^{\alpha\tr UAU^\dag B}=
\sum_\lambda \alpha^{|\lambda|} s_\lambda(\bpow_\infty) s_\lambda(UAU^\dag B)
$$
which gives
$$
\int e^{\alpha\tr UAU^\dag B}d_*U=\sum_\lambda \alpha^{|\lambda|}
\frac{s_\lambda(\bpow_\infty)s_\lambda(A) s_\lambda(B)}{s_\lambda(\mathbb{I}_N)}
=\sum_{\lambda} \frac{\alpha^{|\lambda|}}{(N)_\lambda}s_\lambda(A) s_\lambda(B)
$$
which has the form (\ref{tau(nXY)}). In a similar waym one evaluates
$
\int \left(1-\alpha UAU^\dag B\right)^{-a}d_*U
$
and a number of integrals; see,~for instance, \cite{O-2004-New}.
\er

The expectation values $\mathbb{E}_{0,1}\left\{\bpow_\Delta(AUBU^{-1})\right\}$ and 
$\mathbb{E}_{0,1}\left\{\bpow_{\Delta^a}(AU)\bpow_{\Delta^b}(BU\dag)\right\}$ 
can also be expressed in terms of Hurwitz numbers, but this requires a lot 
of space, and we will not do this; see the last section \cite{NO2020} for details.

\section{Our Models. Products of Random Matrices We Choose}
\vspace{-6pt}

\subsection{Preliminary. On the Products of Random Matrices}

In a number of articles the spectral correlation functions of certain products of were
were~studied. These are products

\be\label{zzdag-zzdag}
 Z_1Z_1^\dag \cdots Z_n Z_n^\dag,
\ee
\be\label{zz-dag}
  Z_1\cdots Z_n(Z_1\cdots Z_n)^\dag,
\ee  
or, a pair of products
\be\label{z,z-dag}
  Z_1\cdots Z_n,\quad (Z_1\cdots Z_n)^\dag;
\ee
see \cite{Ak1,Ak2,AkStrahov}.

Similar products in which complex matrices are replaced by unitary and certain generalizations have also been 
considered \cite{Alexandrov,Chekhov-2014,ChekhovAmbjorn,Chekhov-2016,O-TMP-2017,O-links}.

Below we suggest a generalization of these models (see Examples \ref{E8} and \ref{E6} respectively for 
(\ref{zzdag-zzdag})--(\ref{z,z-dag})). We call it matrix models related to dessin d'enfants,
more precisely, to the so-called clean dessin d'enfants (see \cite{Amburg}).

This is a child's drawing of a "constellation," like the Greek constellation, which is painted in the sky with 
stars, and the sky is a surface with a chosen Euler characteristic.

See below for a more accurate description.

As a part of dessin d'enfants we will need the following modification:

(1) We shall consider mixed ensembles, and in what follows, $X_{\pm i}$ denotes either $Z_{\pm i}$, where
$Z_{-i}=Z_{i}^\dag$, or $U_{\pm i}$, where $U_{-i}=U_i^\dag$.

(2) We need additional data, which we call source matrices $C_{\pm 1},\dots, C_{\pm n}\in \G$. Each random
matrix $X_i$ enters only in combination $X_iC_i$:
\be\label{X-XC}
X_i\,\to\, X_iC_i
\ee 

Such combinations are very helpful (in particular, this makes it possible to consider rectangular random matrices).
We shall use the notion of the \textit{dressed} source matrices 
$C_i,\,i=\pm 1,\dots,\pm n$ with the notation $\Dr_X[*]$:
\be\label{dressing}
\Dr_X[f(C_1,C_{-1}\dots,C_{n},C_{-n})]:=f(X_1C_1,X_{-1}C_{-1}\dots,X_nC_{n},X_{-n}C_{-n})
\ee
where $f$ is any function.
(The letter ${\cal {L}}$ is used to remember
that it is the dressing from the left side.)

The matrices $\{C_i\}$ we call the source matrices which play the role of coupling constants in the matrix models 
below.

\subsection{Models Obtained from Graphs (Geometrical View)\label{geom-view}}

Consider a  connected graph $ \Gamma $ on an orientable connected surface $ \Sigma $ without boundary with Euler
characteristic $ \e $.
We require such properties of the graph: 

(1) its edges do not intersect. For example, the edges of the graph in Figure  \ref{fig:figure1}a do not intersect: the fact is 
that the graph is drawn on a torus, and not on a piece of paper. 

(2) if we cut the surface $\Sigma$ along the edges of the graph, then the surface will decompose into disks (more 
precisely: into pieces homeomorphic to disks) (see Figure \ref{fig:figure5}a,b as an example).

As an example, see Figure \ref{fig:figure1}, which contains all such graphs with two edges.

  \begin{figure}[h] 
\centering
\includegraphics[scale=1]{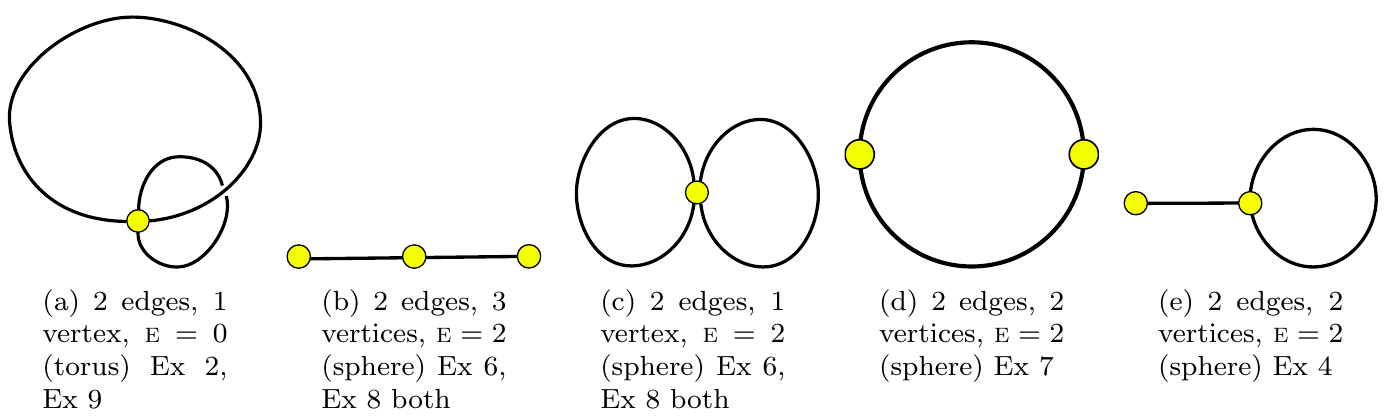}
\caption{All possible graphs with 2 edges (two matrix models). Graph (\textbf{b}) is dual to (\textbf{c}).
Ex means an example below.}
\label{fig:figure1}
\end{figure}

Such a graph is sometimes called a (clean) dessin d'enfants\hl{ (the term dessin 
d'enfants without the additional "clean" serves for such a graph with a bipartite structure), } 
sometimes---a map \cite{ZL}.

Let our graph have $\f$ faces, $n$ edges and $\V$ vertices; then $ \e = \V- n + \f $.

We number all stars (i.e., all vertices of the graph $\Gamma$) with numbers from $1$ to $\V$ and
all faces of $\Gamma$ with numbers from $1$ to $\f$ and all edges of $\Gamma$ with numbers from $1$ to $n$ 
in any way.

We will slightly expand the vertices of the graph and turn them into the \textit{small disk}, which sometimes 
for the sake of visual clarity we will call \textit{stars}.

In this case, the edges coming out of the vertex will divide the
border of a small disk into segments. We orient the boundary of each such segment  clockwise, and the 
segment can be represented 
by \mbox{an arrow} that goes from one edge to another; see the Figures \ref{fig:figure2}a,b,d and
\ref{fig:figure5}a,c as examples.
Our graph will have a total 
of $2n$ arrow segments.

It is more correct to assume that the edges of the graph are very thin ribbons; that is, they have finite 
thickness.  
That is, the edge number $|i|$ has two sides, one of which we number with the number~$i$, and the other with the 
number $-i$ 
(this choice is arbitrary but fixed).
It would be more correct to depict each  edge of  $\Gamma$ numbered by $|i|$
 in the form of a ribbon, the sides of which are two oppositely directed arrows; one arrow has a number $i$, 
 and the second $-i$. However, with the exception of Figure \ref{fig:figure5}, we did not do that so as not to 
 clutter up the drawings.
 \vspace{-12pt}

 \begin{figure}[h] 
\centering
\includegraphics[scale=1]{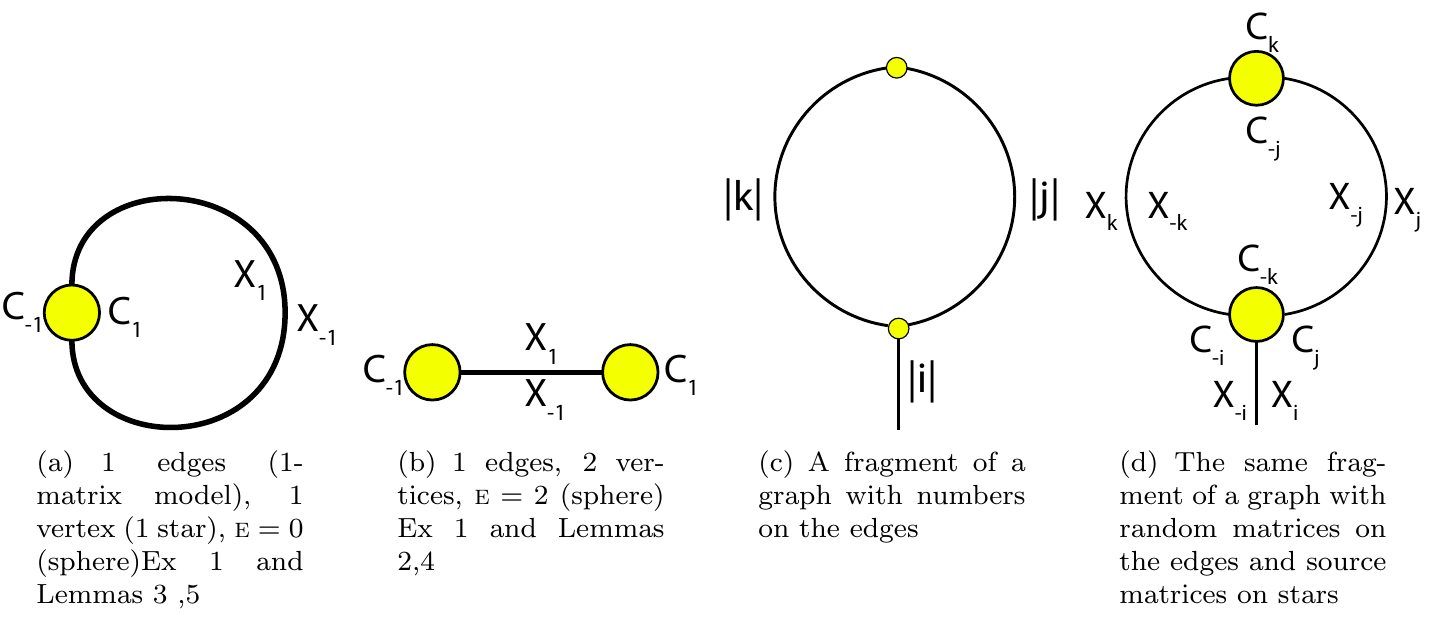}
\caption{Decorated graphs.}
\label{fig:figure2}
\end{figure}
\vspace{-12pt}

\begin{figure}[h] 
\centering
\includegraphics[scale=1]{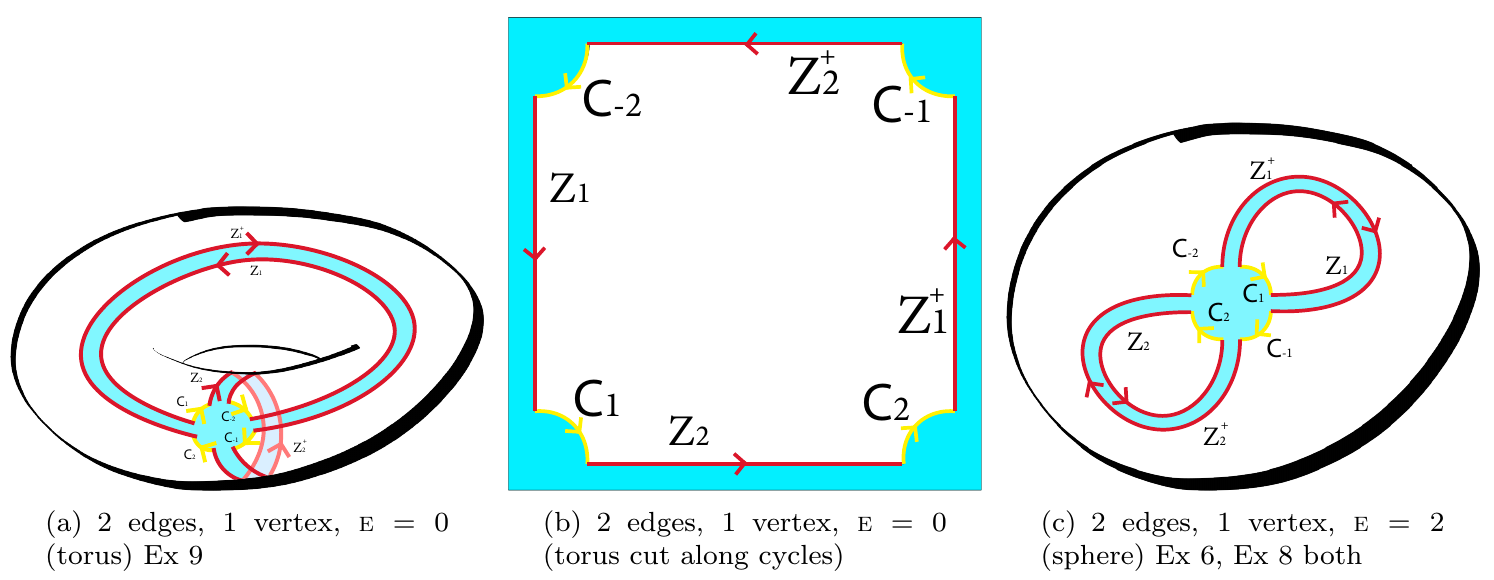}
\caption{Graphs with 2 edges (two matrix models) where (\textbf{a}) is the zoom of (\textbf{a})
in  Figure \ref{fig:figure1} and (\textbf{c}) is the zoom of (\textbf{c}) in  Figure \ref{fig:figure1}.}
\label{fig:figure5}
\end{figure}

Now, each arrow of the side of the edge rests with its end against the arrow-segment---as if continued by the arrow-segment. 
We assign the same number $i$ ($i=\pm 1,\dots,\pm n$) to the side of the edge $ \Gamma $ and to the segment of the small disk, which
continues this side while traversing the face in the positive direction.
  
Additionally, to the number $i$ ($i=\pm 1,\dots,\pm n$) we attribute the product of the $N\times N$ matrices
$$
i\,\to\,X_iC_i
$$
where $ C_i $ is assigned to the $ i $ segment of a small disk and will be called \textit{source matrix},
and $ X_i $ is assigned to the arrow $ i $, which is the side of the edge $ | i | $ of the graph $ \Gamma $;
see Figure \ref{fig:figure2} for an illustration.

Let us use the following numbering. If one side of the edge of the tape $ | i | $ is labeled $ i $, then the other 
side of
the same edges is labeled $ -i $. (It does not matter which side of the edge $|i|$ we assign the number $i$ to, and 
which we assign the number $-i$ to, but we should fix the numbering we have chosen).
The~two sides of an edge of the graph $ \Gamma $ are actually arrows pointing in opposite directions.
(In order not to complicate the drawing, we do not depict these arrows in our figures, except for  Figure \ref{fig:figure5}a,b). 
We~draw the arrows so that they point in the positive 
direction when bypassing the boundary of the face (that is, when bypassing the face in the 
counterclockwise direction).
Each arrow, say $i$ (which~either a positive or a negative number), ends at the boundary of the small disk 
(star) and with further bypassing 
of the face we pass along the segment of the boundary of the star. We assign the same number  $i$ to this~segment.
Acting in this way, we give numbers to all segments of all small disks. It is easy to understand that all the arrows 
on the segments of the boundary of any 
small disc are directed clockwise (i.e., in the negative direction) when passing along the boundary of each small disk.

We attribute a sequential set of numbers to each star as follows.
(this set is defined up to \mbox{a cyclic} permutation, and we will call it a cycle 
associated with the star).
Examples: These are numbers $(1,-1) $ 
assigned to the star (yellow small disk) in  Figure \ref{fig:figure2}a. These are
the numbers $(k,-j)$ that we attribute to the star on the top in the Figure \ref{fig:figure2}c and the numbers $(-k,j,-i)$
that we attribute to the lower star in the same figure.
 Additionally, to each cycle we ascribe the related cyclic products:
$$
(1,-1)\,\leftrightarrow\,C_1C_{-1},\quad 
$$
to the star (yellow small disk) in Figure \ref{fig:figure2}a.  As for the stars in Figure 
\ref{fig:figure2}c we obtain
$$
(k,-j)\,  \leftrightarrow \,C_kC_{-j}\,\quad (-k,j,-i)\,  \leftrightarrow \,C_{-k}C_{j}C_{-i}
$$

Each cycle product we will call the \textit{star's monodromy}. 
As a result, a star's monodromy is a product of matrices that are attributed to 
arrow segments, taken in the same sequence in which arrows follow each other when moving around a small disk clockwise.
We number the stars in any way by numbers from $1$ to $\V$.
The
monodromy of a star $i$ will be denoted by the letter $W_i^*$.

In addition to the edges and in addition to the vertices, we number the 
faces of the graph $\Gamma$ and the corresponding face  monodromies with numbers from $1$ to $\f$.
The cycles corresponding to the face $i$ will be denoted by $f_i$ and defined as follows.
When going around the face border in the positive~direction, namely, counterclockwise 
for an ordinary face (or counterclockwise if the face contains infinity), we~collect segment numbers of small disks; for a face with a number $i$, this ordered collection of ordered 
numbers is $f_i$. 
As in the case of stars, we build cyclic products  $f_i\leftrightarrow W_i$. Examples:
$$
f_1=(1)\,\leftrightarrow\,C_1=W_1 \quad{\rm and}\quad f_2=(-1)\,\leftrightarrow\,C_{-1}=W_2 
$$
for two faces in Figure  \ref{fig:figure2}a.

We also introduce \textit{dressed} monodromies:
$$
{\cal L}_X(W_1)=X_1C_1,\quad {\cal L}_X(W_{-1})=X_{-1}C_{-1}
$$

Additionally, for the face in Figure \ref{fig:figure2}c:
$$
f_1= (-j,-k)  \, \leftrightarrow \,C_{-j}C_{-k}=W_1\quad \leftrightarrow \Dr_X(W_1)=
X_{-j}C_{-j}X_{-k}C_{-k}
$$

Thus, we have two sets of cycles and two sets of monodromies: vertex cycles and vertex~monodromies
\be
\sigma_i^{-1}\,\leftrightarrow \, W_i^*,\quad i=1,\dots,\V
\ee
and face cycles and face monodromies:
The cycles corresponding to the face $i$ will be denoted by $f_i$; we have
\be\label{face-monodromy}
f_i\,\leftrightarrow \, W_i \,\leftrightarrow \, \Dr_X[W_i],\quad i=1,\dots,\f
\ee

Let us note that both cycles and monodromies are defined up to the cyclic permutation.

\br
Important remark. 
Please note that each of the matrices $C_i,\,i=\pm 1,\dots,\pm n$ enters the set of monodromies 
$W_1,\dots,W_f$ once and only once. Accordingly, each random matrix $X_i,\,i=\pm 1,\dots,\pm n$ is included once 
and only once in the set of dressed monodromies $\Dr_X[W_1],\dots,\Dr_X[W_f]$. This determines the class of matrix 
models  that we will consider
and which we will call matrix models of dessin d'enfants.

\er

\br
In the description of maps (the same, of clean dessin d'enfants) the following combinatorial relation is well-known;
see the wonderful texbook \cite{ZL}, Remark 1.3.19:
\be\label{graph-comb}
\left( \prod_{i=1}^n \alpha_i \right)\left(\prod_{i=1}^\f f_i  \right)=\left(\prod_{i=1}^\V \sigma_i^{-1}  \right)
\ee
where each of $\alpha_i$, $f_i$ and $\sigma_i$ belongs to the permutation group $S_{2n}$. Here $f_i$ are face cycles,
$\sigma_i$ are vertex cycles and~$\left( \prod_{i=1}^n \alpha_i \right)$ is an involution without fixed points;
each $\alpha_i$ transposes $i$ and $-i$. Since $\alpha_i^2=1$ we can also~write
\be\label{graph-comb*}
\left( \prod_{i=1}^n \alpha_i \right)\left(\prod_{i=1}^\V \sigma_i^{-1}  \right)=\left(\prod_{i=1}^\f f_i  \right)
\ee

The reader can check this relation for the pair of the simplest examples
where $n=1$.

For any given set of the face cycles $ f_1, \dots, f_\f $, the relation (\ref{graph-comb}) allows you to 
  \textit{uniquely reconstruct} the set of the vertex
cycles $ \sigma_1^{- 1}, \dots, \sigma_V^{-1} $ (the power $-1$ is related to the negative (clockwise) counting
of numbers).
\er

We get 
\bp
\label{simplest-intgral}

\be\label{d=1}
  \hbar^{-n_1}\E_{n_1,n_2}\left\{ \prod_{i=1}^\f  \Dr_X\left[\tr W_i\right]\right\}=N^{-n_2}\prod_{i=1}^{\V} \tr W^*_i 
\ee
and
\be\label{d=1*}
  \hbar^{-n_1}\E_{n_1,n_2}\left\{ \prod_{i=1}^\V  \Dr_X\left[\tr W_i^* \right]\right\}=N^{-n_2}\prod_{i=1}^{\f} \tr W_i \,,
\ee
where we use the notation (\ref{dressing}) and where the parameters $n_1,n_2,\hbar$ were defined in 
Section \ref{MixedEnsembles}.
\ep

We note that each trace is a sum of monomials.
Let us note that each of $X_i,\,i=\pm 1,\dots,\pm n $ enters only once in each monomials term insider
of the integral in the left and side of (\ref{d=1}) and of (\ref{d=1*}).

There are two ways to prove these relations: algebraic and geometrical ones. 
In short, the sketches of the proof are are as follows. 
We should
notice that $\bpow_{(1)}(*)=s_{(1)}(*)$
and apply Lemmas \ref{Lemma-s-Z}--\ref{Lemma-s-s-U} while
having in mind that this is actually the cut-or-join procedure (\ref{duality}) stated below in Section \ref{alg}.
In this case 
the difference between Lemmas \ref{Lemma-s-Z}, \ref{Lemma-s-s-Z} and  Lemmas \ref{Lemma-s-U}, \ref{Lemma-s-s-U}
is only in the power of the factor $N$, since $(N)_\lambda=N$ in case $\lambda=(1)$.

Geometric way: We use the relation
$$
\hbar^{-1}\E_{n_1,0}
\left\{ (Z_i)_{a,b}(Z_j^\dag)_{b',a'}  \right\}= \delta_{i,j}\delta_{a,a'}\delta_{b,b'}
$$
and calculate monodromies along faces of the graph $\Gamma$ as described in the beginning of this section.

\br
From geometrical point of view in this way we create a surface by gluing the polygons related to the dressed 
face monodromies; see \cite{NO2020}.
\er

\br
You may notice some similarities between formulas (\ref{d=1}), (\ref{d=1*}) and formulas 
(\ref{graph-comb}), (\ref{graph-comb*}). Additionally, there is. One can say that the role of involution $\alpha_i$ is played by 
the integration  over the matrix $X_i$ (by the Wick coupling of $X_i$ and $X_{-i}$).
\er

We draw attention to two facts.
 \begin{itemize}
  \item The answer (i.e., the left-hand side of \eqref{d=1}) depends only on the spectrum of star monodromies
  $$
  {\rm Spect}(W_1^*),\dots,{\rm Spect}(W_\V^*).
  $$
  \item The answer does not depend on how exactly in our model we distribute the matrices 
$\{Z\}$ and $\{U\}$ to dress the source matrices---only two numbers $n_1$ and $n_2$ are important. 
For example, it~does not matter, in the right-hand side of (\ref{d=1}),
whether we dress $\Dr_X[C_i,C_{-i}]=U_iC_i,U_{-i}C_{-i}$ and $\Dr_X[C_j,C_{-j}]=Z_jC_j,Z_{-j}C_{-j}$ or 
$\Dr_X[C_i,C_{-i}]=Z_iC_i,Z_{-i}C_{-i}$ and $\Dr_X[C_j,C_{-j}]=U_jC_j,U_{-j}C_{-j}$.
  
 \end{itemize}

\subsection{Our Models. Algebraic View \label{alg}}

You can forget about graphs and dessin d'enfants and set them out differently.

We have a set of random matrices 
$$
X_{\pm 1},\dots, X_{\pm n},\quad X_{-i}=X_i^\dag
$$
and there are as many source matrices
\be\label{alphabet}
C_{\pm 1},\dots, C_{\pm n},
\ee
and we want to consider expectation value of 
 the products of these matrices. We want each matrix $X_i$ to come in combination with  
the source matrix labeled with the same $i$; see (\ref{X-XC}) and (\ref{dressing}).

The set (\ref{alphabet}) we will call the alphabet of pairs
and the matrices $\{C_i\}$ can be considered letters of the alphabet.
 
The products of these matrices can be considered as words constructed from letters of the alphabet of pairs. 
Consider a group of words
$$
W_1,\dots,W_\f
$$
with the following condition: each of matrices $C_i,\,i=\pm 1,\dots,\pm n$ is included \textit{once and only once} in 
one of the words of this group.

In addition, we ask the following property: In this group of words there is no such subset of words that could be 
constructed from the alphabet of pairs with a smaller set of pairs of letters (in~other words, we will consider
only connected groups of words. The connected group of words will be related to the connected graphs $\Gamma$).

Each word $W_i$ can be associated with an ordered set of numbers $f_i$ consisting of
 the numbers of the matrices included in the word: $f_i\leftrightarrow W_i$. Recall that each word
  $W_i$ (and respectively $f_i$) is defined up to a cyclic permutation.

 There is a one-to-one correspondence between dessin d'enfants with $n$ edges and $\f$ faces and word sets 
 constructed in this way. Based on the set $W_1,\dots,W_\f$, the surface $\Sigma$ is built uniquely. To~do~this, 
 the procedure is~as follows. First, for each  word, say $W_i$, we associate the polygon with the edges, numbered by the numbers of the 
 matrices  in the product, namely, by the set $f_i$.
By going around the boundary of the polygon counterclockwise, we assign the numbers of the matrices to the edges 
if we write the word from left to right.
We get a set of $\f$ polygons. After that, we glue these polygons so that the side with the number $i$ sticks together 
with the side with the number $-i$.
We assign an orientation to each polygon, considering its edges arrows, showing the direction of counterclockwise. 
We glue the edges so that the beginning of the arrow $i$ sticks together with the end of the arrow $-i$.
We get the surface $ \Sigma $ and on it is the graph $ \Gamma $, whose edges are glued from two oppositely directed 
arrows (ribbon edge).

To determine the Euler characteristic of $\Sigma$ we need to know the number of vertices of the graph $\Gamma$.

{Cut-or-join procedure.}
Purely algebraically, one should act like this.

Consider $ W_1 \otimes W_2 \otimes \cdots \otimes W_\f $
(the order in this tensor product is not important) and the set of involutions $ T_i, \, i = 1, \dots, n $, which act on this tensor product as follows.
Each involution of $ T_i $ does not affect those $ W_a $ that contain neither $ C_i $ nor $ C_{- i} $. Two situations are possible. (I) The matrices $ C_i $ and $ C_{- i} $ are in the same word, say, the word $ W_a $. How then can we
to rearrange the matrices with the word cyclically, to bring it to the form $ C_i X C_{- i} Y $, where $ X $ and 
$ Y $ are some matrices. (II) The matrices $ C_i $ and $ C_{- i} $ are included in different words; in this case 
we will write these two words as $ C_iX $ and $ C_{- i} Y $.
As~you can see, the action of involutions $T_i$ corresponds to taking the integral in Lemmas
\ref{Lemma-s-Z}--\ref{Lemma-s-s-U}.
Then
\bea
T_i\left[ \cdots \otimes C_i X C_{-i}Y \otimes \cdots \right]=\cdots \otimes C_iX \otimes C_i Y \otimes \cdots
\\
T_i\left[ \cdots  \otimes C_iX \otimes C_i Y \otimes \cdots \right] =\cdots \otimes C_i X C_{-i}Y \otimes \cdots 
\eea

It is easy to see that involutions commute: $ T_i [T_j [*]] = T_j [T_i [*]] $.

The transformation 
\be\label{duality}
W_1,\dots, W_\f\, \leftrightarrow \, W_1^*,\dots, W_\V^*
\ee
can be obtained purely algebraically in $n$ steps. We call these two sets of matrices \textit{dual sets}.

Proposition:
\bea\label{cut-or-join}
\prod_{i=1}^n T_i\left[ W_1\otimes W_2 \otimes \cdots \otimes W_\f \right]=
W_1^*\otimes \cdots \otimes W_\V^*
\\
\prod_{i=1}^n T_i\left[ W_1^*\otimes \cdots \otimes W_\V^* \right]=W_1\otimes W_2 \otimes \cdots \otimes W_\f
\eea

\br
Actually this is a manifestation of the Equation (\ref{graph-comb}) where $\alpha_i$ is related to $T_i$.
An involution without fixed points $ \prod_{i = 1}^n T_i $ takes the graph $ \Gamma $ to the graph dual to it.
\er

Here are some examples (\ref{duality}) with explanations about the geometric representation in Figures \ref{fig:figure1}--\ref{fig:figure4}
and in some other graphs:

\bx\label{1} Suppose $W_1=C_1$, $W_2=C_{-1}$. In 1 step we get:
$$
 W_1\otimes W_{2}=C_1\otimes C_{-1} \leftrightarrow   C_1C_{-1} =W_1^*
$$

Therefore $\V=1$ and as we have $\f=2,n=1$; then the Euler characteristic is $\e=\f-n+\V=2$.
\mbox{The right} hand side can be obtained from Figure \ref{fig:figure2}a as the face monodromies while
the left-hand side can be obtained as star monodromies there.
The left-hand side contains also the face monodromies in Figure \ref{fig:figure2}b; and
the right-hand side can be obtained as star monodromies there. (This is because graphs (a) and (b) are dual ones.)
\ex
\bx\label{2}
In 2 steps:
$$
 W_1= C_1C_2C_{-1}C_{-2}\leftrightarrow  C_1C_{-2}C_{-1}C_{2} = W_1^*
$$

Therefore $\V=1$, and as we have $\f=1,n=2$, then the Euler characteristic is $\e=\f-n+\V=0$ (torus).
The~right-hand side can be obtained from Figure \ref{fig:figure1}a as the face monodromies while
the left-hand side can be obtained as star monodromies there; see the more detailed Figure \ref{fig:figure5}a.
The dual graph looks like the same as the graph on the torus: there is 1 vertex and 1 face; the left-hand side
plays the role of the face monodromy for the dual graph. This a particular case of Ex 9.
\ex
\bx\label{3} 
In two steps:
$$
W_1\otimes W_2\otimes W_3=
C_{-1}C_{-2}\otimes C_1\otimes C_{2}\leftrightarrow   C_1C_{-1}C_{2}C_{-2} =W_1^*
$$

Therefore $\V=1$ and as we have $\f=3,n=2$,  the Euler characteristic is $\e=\f-n+\V=2$
The right-hand side can be obtained from Figure \ref{fig:figure1}b as the face monodromies while
the left-hand side can be obtained as star monodromies there. 
The left-hand side contains also the face monodromies in Figure \ref{fig:figure1}c (see  Figure \ref{fig:figure5}c for more details);
the right-hand side can be obtained as star monodromies there. (This is because graphs \ref{fig:figure1}b and \ref{fig:figure1}c are dual ones).
\ex

\bx\label{4}   In 2 steps: 
$$
  C_{1}C_2C_{-1}\otimes C_{-2} \leftrightarrow   C_1C_{-2}C_{2}\otimes C_{-1}
$$

Therefore $\V=2$ and as we have $\f=2,n=2$,  the Euler characteristic is $\e=\f-n+\V=2$.
The right-hand side can be obtained from Figure \ref{fig:figure1}e as the face monodromies while
the left-hand side can be obtained as star monodromies there.
The left-hand side contains also the face monodromies of the dual graph which is the graph
of the same type (the loop from which the segment sticks out).

\bx\label{5}
 In 5 steps: 
 \begingroup\makeatletter\def\f@size{8}\check@mathfonts
\def\maketag@@@#1{\hbox{\m@th\fontsize{10}{10}\selectfont \normalfont#1}}%
$$
W_1\otimes W_2\otimes W_3=
C_1C_2C_3C_4\otimes C_{-3}C_{-2}C_5\otimes C_{-5}C_{-1}C_{-4}\leftrightarrow  
C_2C_{-2}C_{-5}\otimes C_{2}C_{-3}\otimes C_5C_3C_{-4} \otimes C_4C_{-1}
=W_1^*\otimes W_2^*\otimes W_3^*\otimes W_4^*
$$
\endgroup

Therefore $\V=4$ and as we have $\f=3,n=5$,  the Euler characteristic is $\e=\f-n+\V=2$
(the same is obtained from Figure \ref{fig:figure4}d).
\ex

\ex
\bx\label{6}
  In $n$ steps:
$$
C_{1}C_2\cdots C_{n}C_{-n} \cdots C_{-1} \leftrightarrow
C_{1}C_{-2}\otimes C_2 C_{-3}\otimes \cdots \otimes C_{n-1}C_{-n}\otimes C_n \otimes C_{-1}
$$

Therefore $\V=n+1$ and as we have $\f=1$,  the Euler characteristic is $\e=\f-n+\V=2$
The left-hand side is related to $\Gamma$ in form of a chain with $n+1$ vertices, $n$ edges drawn on the sphere
(see Figure \ref{fig:figure1}b for the $n=2$ chain and see Figure \ref{fig:figure3}a for the $n=3$ chain).
In case each source matrix is the identity ones, this is related 
to (\ref{zz-dag}) which is a rather popular product. 
The right side is represented by such a graph: there are $n$ circles; each subsequent one decreases and is inside 
the previous one. They  all touch at one point (the vertex); see Figure \ref{fig:figure3}b.
\ex
\bx\label{7}
  In $n$ steps: 
$$
  C_{1}C_2\cdots C_n\otimes C_{-n}C_{1-n}\cdots C_{-1} \leftrightarrow   
  C_{1}C_{-2}\otimes C_2 C_{-3}\otimes \cdots \otimes C_{n-1}C_{-n}\otimes C_n C_{-1}
$$

Therefore $\V=n$, and as we have $\f=2$, then the Euler characteristic is $\e=\f-n+\V=2$.
(the case $n=2$ is obtained from Figure \ref{fig:figure1}d). The left-hand side is related to $\Gamma$ in form of 
the chain from the previous case
where we replace $n$ by $n-1$ and connect its ends by $n$-th edge. Namely, it is $n$-gon drawn on the sphere.
The right-hand side is related to the graph where two vertices are connected by $n$ edges.
In case $n=4$,  the right-hand side can be obtained from Figure \ref{fig:figure3}d as the face monodromies while
the left-hand side can be obtained as star monodromies there.
The left-hand side contains also the face monodromies in Figure \ref{fig:figure3}c;
the right-hand side can be obtained as star monodromies there. (Indeed, graphs in
Figure \ref{fig:figure3}d and in Figure \ref{fig:figure3}c are dual ones.)
\ex
\bx\label{8}
  In $n$ steps:
$$
 C_{1}C_{-1}C_2C_{-2}\cdots C_{n}C_{-n}
\leftrightarrow   C_{-1}C_{-2}\cdots C_{-n}\otimes C_1\otimes C_2\otimes \cdots \otimes C_n 
$$

Therefore $\V=n+1$, and as we have $\f=1$, then the Euler characteristic is $\e=\f-n+\V=2$
 The left-hand side is related to a star-graph $\Gamma$; the right-hand side---to a petal graph (as an example take $n=3$ and look at Figure \ref{fig:figure4}a and dual Figure \ref{fig:figure4}b.
   \begin{figure}[h] 
\centering
\includegraphics[scale=1]{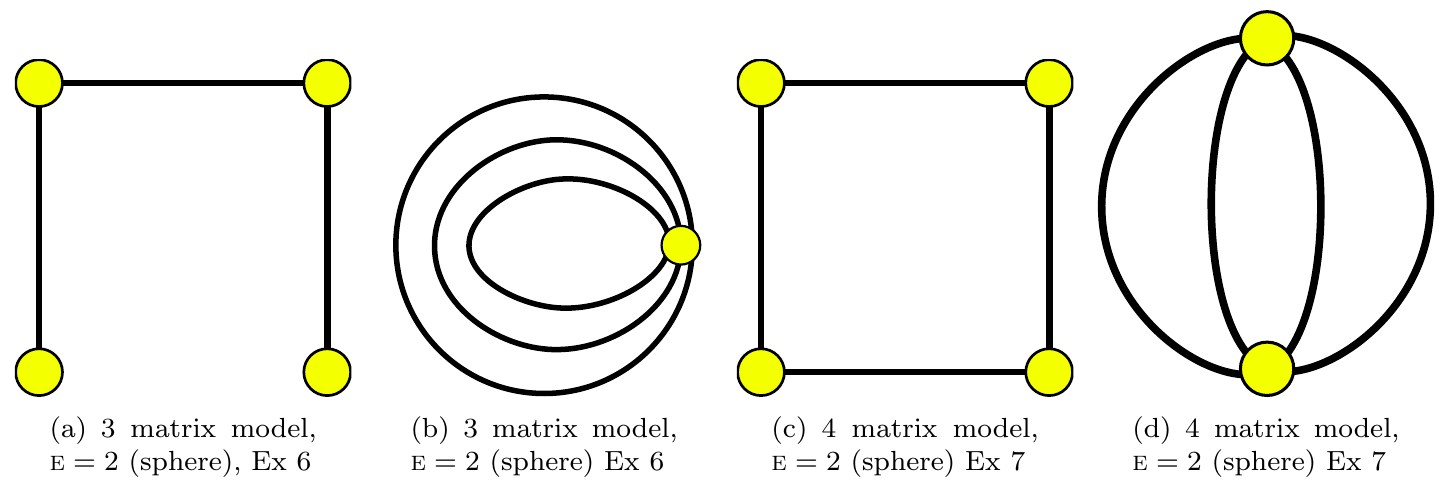}
\caption{Graphs drawn without decoration. 
Graph (\textbf{a}) is dual to (\textbf{b}). Graph (\textbf{c}) is dual to (\textbf{d}).}
\label{fig:figure3}
\end{figure}
\vspace{-12pt}

  \begin{figure}[h] 
\centering
\includegraphics[scale=1]{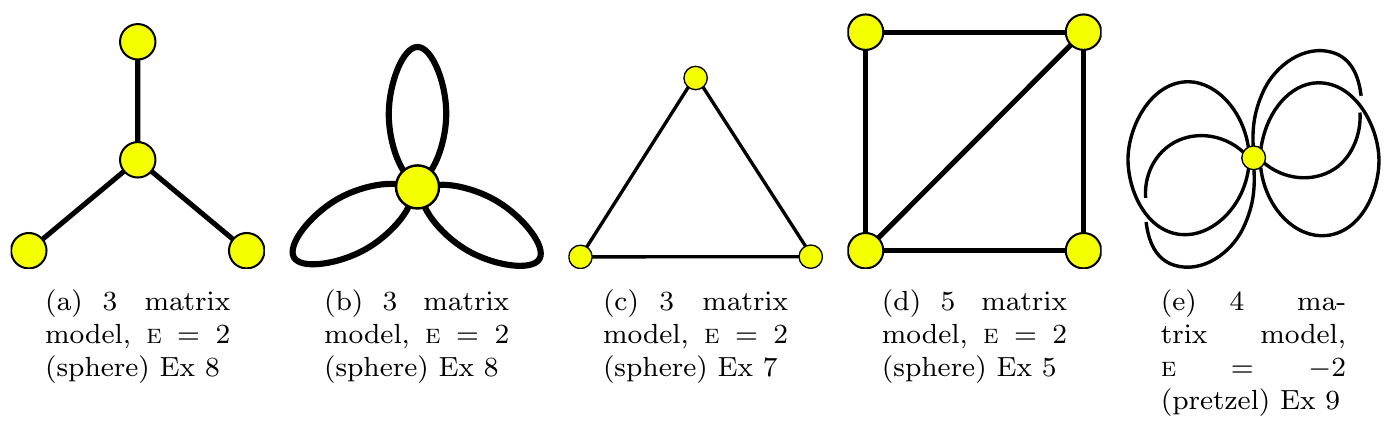}
\caption{Graphs drawn without decoration. Graph (\textbf{a}) is dual to (\textbf{b}).}
\label{fig:figure4}
\end{figure}
 After the dressing procedure, in case all source 
 matrices are chosen to be identical ones we get a popular product (\ref{zzdag-zzdag}).

\ex
\bx\label{9}
 Suppose $W_1=C_{a_1}C_{b_1}C_{-a_1}C_{-b_1}\cdots  C_{a_g}C_{b_g}C_{-a_g}C_{-b_g}$,
where we number source matrices according~$a_i,b_i$-cycles structure of a Riemann surface of genus $g$.
In $n=2g$ steps we obtain 
$$
C_{a_1}C_{b_1}C_{-a_1}C_{-b_1}\cdots  C_{a_g}C_{b_g}C_{-a_g}C_{-b_g}\leftrightarrow 
C_{-a_1}C_{b_1}C_{a_1}C_{-b_1}\cdots  C_{-a_g}C_{b_g}C_{a_g}C_{-b_g}
$$

Therefore $\V=1$, and as we have $\f=1$, $n=2g$, then the Euler characteristic is $\e=\f-n+\V=2-2g$

The case $n=2$ yields a torus and was considered in Example \ref{2}. The case $n=4$ can be related to a
graph drawn on a pretzel; see Figure \ref{fig:figure4}e. The left-hand side and the right-hand sides
are related to dual graphs each of which has one face and one vertex and is drawn on a surface with genus $g$
whose edges are $a_i,b_i$ cycles.

\ex
\bx\label{10}
 In $n=6$ steps we obtain 
$$
C_{-1}C_{-2}C_{-3}\otimes C_{-5}C_{-4}C_{1}\otimes C_{-6}C_{5}C_{2} \otimes C_{6}C_{4}C_{3} \leftrightarrow 
C_{1}C_{4}C_{-3}\otimes C_{-5}C_{-1}C_{2}\otimes C_{-6}C_{-2}C_{3} \otimes C_{-4}C_{5}C_{6}
$$

Therefore $\V=4$, and as we have $\f=4,n=6$, then the Euler characteristic is $\e=\f-n+\V=2$.
This can be represented as a tetrahedron inscribed in a sphere. The tetrahedron graph is self-dual.

\ex

\section{Expectation Values of Matrix Products}

Lemmas \ref{Lemma-s-Z}--\ref{Lemma-s-s-U}
are generalized to the case of mixed ensembles; see 
Propositions \ref{prop-Schur} and \ref{prop-power} below.

Propositions \ref{prop-Schur}-\ref{prop-power} are extended versions of statements 
studied in \cite{NO2019,NO2020,NO2020F}.

\bp\label{prop-Schur}
\label{Schur-expectation} Consider ensemble (\ref{mixed}).
Consider dual sets $W_1,\dots, W_\f$ and $W^*_1,\dots,W^*_\V$ of (\ref{duality}).
For any given set of partitions $\lambda^1,\lambda^2,\dots,\lambda^\f =\lambda\in \Upsilon_d$, we have
\begin{myequation}
\label{E-Schur}
 \E_{n_1,n_2}\left\{ \Dr_X\left[ s_{\lambda^1}\left( W_1 \right) \cdots
 s_{\lambda^\f}\left( W_\f \right)\right]\right\}  =
\delta_{\lambda^1,\lambda^2,\dots,\lambda^\f} \hbar^{n_1d}
\left(s_\lambda(\bpow_\infty)   \right)^{-n_1} 
\left(s_\lambda(\mathbb{I}_N)   \right)^{-n_2} 
s_{\lambda}\left(W_1^*\right)\cdots s_{\lambda}\left(W_\V^*\right),
\end{myequation}
  where  $ \delta_{\lambda^1,\lambda^2,\dots,\lambda^\f} $ is equal to 1 in case 
$\lambda^1=\lambda^2=\cdots =\lambda^\f$ and to 0 otherwise.

Similarly, for any set of partitions $\lambda^1,\lambda^2,\dots,\lambda^\V =\lambda\in \Upsilon_d$, we get
\begin{myequation}
\label{E-Schur1}
 \E_{n_1,n_2}\left\{ \Dr_X\left[ s_{\lambda^1}\left( W^*_1 \right) \cdots
 s_{\lambda^\f}\left( W^*_\V \right)\right]\right\}  =
\delta_{\lambda^1,\lambda^2,\dots,\lambda^\V} \hbar^{n_1d}
\left(s_\lambda(\bpow_\infty)   \right)^{-n_1} 
\left(s_\lambda(\mathbb{I}_N)   \right)^{-n_2} 
s_{\lambda}\left(W_1\right)\cdots s_{\lambda}\left(W_\f\right),
\end{myequation}
  where  $ \delta_{\lambda^1,\lambda^2,\dots,\lambda^\f} $ is equal to 1 in case 
$\lambda^1=\lambda^2=\cdots =\lambda^\f$ and to 0 otherwise.

\ep

The sketch of proof. The proof is based on the cut-or-join procedure (\ref{duality}) of Section \ref{alg} which
is the result of the step-by-step application of Lemmas \ref{Lemma-s-Z}--\ref{Lemma-s-s-U}.
 The different (geometrical) proof is based on the treating of the Wick rule as a way to glue surfaces from
 polygons and the use of (\ref{orth2}) and of (\ref{Mednyh}), (\ref{Mednyh-m}).

\begin{Corollary}
 
\label{prop-Schur-det}
\label{Schur-expectation-det} Consider ensemble (\ref{mixed}).
Consider dual sets $W_1,\dots, W_\f$ and $W^*_1,\dots,W^*_\V$ of (\ref{duality}), such~that~$\det W_i,\det W_j^*\neq 0$ for $i=1,\dots,\f$, $j=1,\dots,\V$.
Suppose that $\alpha_m=\max\left( \alpha^1,\dots,\alpha^\f \right)=:\alpha$, where $\alpha^1,\dots,\alpha^\f $
is a given set of nonnegative integers. Consider a given set of partitions $\lambda^i,\,i=1,\dots,\f$ and
denote $\lambda=\lambda^{(m)}$.
We have
\be\label{E-Schur-det}
 \E_{n_1,n_2}\left\{ \Dr_X\left[ \prod_{i=1}^{\f} s_{\lambda^i}\left( W_i \right)
 \det\left((W_i )^{\alpha_i}\right) \right]\right\}  =
 \ee
 \be
\delta_{\lambda^1+\alpha_1,\lambda^2+\alpha_2,\dots,\lambda^\f+\alpha_f} 
\left(\hbar^{-(|\lambda|+\alpha N)} s_\lambda(\bpow_\infty)    \frac{(N)_\lambda}{(N+\alpha)_\lambda}\right)^{-n_1}
\left(s_\lambda(\mathbb{I}_N)   \right)^{-n_2} 
\prod_{i=1}^\V
s_{\lambda}\left(W_i^*\right)\det\left((W_i^*)^\alpha  \right)
\ee
  where  $ \delta_{\lambda^1+\alpha_1,\lambda^2+\alpha_2,\dots,\lambda^\f+\alpha_\f} $ is equal to 1 in case 
for each $k=1,\dots,N$ we have $\lambda^{(1)}_k+\alpha_1=\lambda^{(2)}_k+\alpha_2=\cdots =\lambda^{(\f)}_k+\alpha_\f$,
and it is 0 otherwise.
\end{Corollary}

\begin{Proposition}[\cite{NO2020F}]\label{prop3_dv}
let $\Delta^i=\left(\Delta^i_1,\Delta^i_2,\dots \right),\,i=1,\dots,k$ be a set 
of partitions of the weights $d_1,\dots,d_k=d$.

Let $\mu^{k+1}=(\mu^{(k+1)}_1,\mu^{(k+1)}_2,\dots),\dots,\mu^{(\f)}=\mu=(\mu_1,\mu_2,\dots)$,
where $1<k<\f$, be a set of partitions. We get
\be
\label{skew1}
\E_{n_1,n_2}\left\{ 
\Dr_X\left[  \prod_{i=1}^k \bpow_{{\Delta}^i}(W_i)
\prod_{i=k+1}^{\f} s_{\mu^i}(W_i)\right]
  \right\}=
\ee
\be
\delta_{|\mu|,d}
\delta_{d_1,\dots, d_k}
\delta_{\mu^{k+1},\dots,\mu^{\f-k} } \hbar^{n_1d}\left((N)_\mu\right)^{-n_2}
 \left( \frac{{\rm dim}\,\mu}{d!} \right)^{-n_1-n_2}
 \chi_\mu({\Delta}^1)\cdots \chi_\mu({\Delta}^k)
\prod_{i=1}^\V s_\lambda(W_i^*)
\ee
where $(N)_\mu$  is given by (\ref{N-lambda}) and $\chi_\mu(\Delta)$
is the character of the symmetric group; see Remark (\ref{char-char}).
The symbol $\delta_{\mu^{k+1},\dots,\mu^{\f} }$ is equal to 1 in case $\mu^{k+1}=\cdots=\mu^{\f-k} $
and is equal to 0 otherwise.

Similarly, let $\mu^{k+1},\dots,\mu^{\V}=\mu$, where $1<k<\V$, be a set of partitions. Then
\be
\label{MM1-all-skew}
\E_{n_1,n_2}\left\{ 
\Dr_X\left[  \prod_{i=1}^k \bpow_{{\Delta}^i}(W_i^*)
\prod_{i=k+1}^{\V} s_{\mu^i}(W_i^*)\right]
  \right\}=
\ee
\be
\delta_{|\mu|,d}
\delta_{d_1,\dots, d_k}
\delta_{\mu^{k+1},\dots,\mu^{\V-k} } \hbar^{n_1d}\left((N)_\mu\right)^{-n_2}
 \left( \frac{{\rm dim}\,\mu}{d!} \right)^{-n_1-n_2}
 \chi_\mu({\Delta}^1)\cdots \chi_\mu({\Delta}^k)
 \prod_{i=1}^\f s_\lambda(W_i) 
\ee

\end{Proposition}

The Proposition is derived from the previous one using (\ref{char-char}) and (\ref{orth2}).

\bp\label{prop-power} 
Let $\Delta^i=\left(\Delta^i_1,\Delta^i_2,\dots \right),\,i=1,\dots,\f$ be a set 
of partitions of weights $d_1,\dots,d_\f=d$.

Then
\bea\label{E-power1}
 &\E_{n_1,n_2}\left\{ \Dr_X\left[ \bpow_{\Delta^1}\left( W_1 \right) \cdots
 \bpow_{\Delta^\f}\left( W_\f \right)\right]\right\}
 \prod_{i=1}^\f \frac{1}{z_{\Delta^i}}\,
 =
 \\   \label{E-power2}
   &\delta_{d_1,\dots, d_\f}           \hbar^{n_1d}
\sum\limits_{\widetilde{\Delta}^1,\dots,\widetilde{\Delta}^\V\in \Upsilon_d}
\bpow_{\tilde\Delta^1}(W^*_{1})\cdots
\bpow_{\tilde\Delta^\V}(W^*_{\V})
H_{\e}( \widetilde{\Delta}^1,\dots, \widetilde{\Delta}^\V ,\Delta^1,\dots, \Delta^\f|n_2),
\eea
where $\delta_{d_1,\dots, d_\f} =1$ in case $d_1=\cdots = d_\f$ and 0 otherwise, and $\Upsilon_d$\
is the set of all partitions of weight $d$.
The~prefactor $H_{\e}( \widetilde{\Delta}^1,\dots, \widetilde{\Delta}^\V ,\Delta^1,\dots, \Delta^\f |n_2)$
in (\ref{E-power2})
is the weighted Hurwitz number (\ref{Mednyh-m}) with $k=\f+\V$ and $\e=\f-n+\V$ .

Similarly, for a given set of partitions $\tilde{\Delta}^i=\left(\tilde{\Delta}^i_1,
\tilde{\Delta}^i_2,\dots \right),\,i=1,\dots,\V$ with weights~$d_1,\dots,d_\V=d$~respectively, we get
\bea\label{E-power1*}
 &\E_{n_1,n_2}\left\{ \Dr_X\left[ \bpow_{\tilde{\Delta}^1}\left( W^*_1 \right) \cdots
 \bpow_{\tilde{\Delta}^\V}\left( W^*_\V \right)\right]\right\}=
 \\   \label{E-power2*}
   &\delta_{d_1,\dots, d_\V}           \hbar^{n_1d}
\sum\limits_{\widetilde{\Delta}^1,\dots,\widetilde{\Delta}^\f\in \Upsilon_d}
\bpow_{\Delta^1}(W_{1})\cdots
\bpow_{\Delta^\V}(W_{\V})
H_{\e}( \widetilde{\Delta}^1,\dots, \widetilde{\Delta}^\V ,\Delta^1,\dots, \Delta^\f|n_2),
\eea
where $H_{\e}( \widetilde{\Delta}^1,\dots, \widetilde{\Delta}^\V ,\Delta^1,\dots, \Delta^\f)$ is 
exactly the same as in (\ref{E-power2}).
\ep
Propositions \ref{prop-Schur}--\ref{prop-power} and  are equivalent. This can be proven with the help of 
(\ref{orth1}), (\ref{orth2}) and (\ref{Mednyh-m}). 

\br  Proposition \ref{prop-power} was proved  in \cite{NO2020} using a geometrical construction of
Hurwitz numbers as a number of ways to glue polygons. Each matrix entry,
say $(X_i)_{a,b}$, may be drawn as an arrow with labels $a$ and $b$ at the startpoint and the endpoint 
respectively. We draw solid arrow for an entry of a random matrix and a dashed arrow for an entry of a source matrix. 
The product of matrices we draw as arrows sequentially assigned to each other. The trace of a product is drawn 
as a polygon. Now each $L_X\left[\ttr W_c\right]$ is a polygon with alternating solid and dashed-edge arrows;
we orient the edges counterclockwise. In \cite{NO2020} we named such polygons countries. 
Thus, we relate each dressed word to a country.
It may be shown that the expectation
$$
\E_{n,0}\left\{ L_X\left[ \ttr\left( W_1 \right)\cdots \ttr\left( W_1 \right)\right]\right\}
$$
may be viewed as the result of gluing of the net of countries into a surface, say $\Sigma$, and $\e=\f-n+\V$ being
the Euler characteristic of this surface. This is a result of the Gaussian integration over matrices $Z_i$.
Oppositely~directed solid arrows (corresponding to $Z_i$ and $Z^\dag_i$) form sides of ribbon edges. 
These edges end at a boundary of disks (inflated vertices)
with dashed boundaries (source matrices are attached to the segments of dashed boundaries). In \cite{NO2020} we named this disk
watchtowers. There are 
$n$ ribbon edges (the boarders of countries) and $2n$ dashed edges (segments of boundaries of disks---of 
the boarders of the watchtowers); there are $\V$ watchtowers; and there are 
$\f$ countries with alternating (solid-dashed) edges. There are $2n$ 3-valent vertices (ends of ribbons):
there are two dashed arrows (one is outgoing; another is incoming) and one ribbon (one side is the solid outgoing
arrow; the other side is a incoming solid arrow) attached to each vertex.
This is a graph $\Gamma$ drawn on $\Sigma$. This graph is related~to
\be
\E_{n,0}\left\{ L_X\left[ \ttr\left( W_1 \right)\cdots \ttr\left( W_1 \right)\right]\right\}=
  \hbar^{n}
\ttr\left(W^*_{1}\right)\cdots
\ttr\left(W^*_{\V}\right),
\ee
which is (\ref{E-power1}) for $d=1$ (in this case, all $\Delta^i$ has weight 1, and $p_{(1)}(W)=\ttr(W)$).
In case $d>1$, instead of each $\ttr\left(W_i\right)$ we have the product
$$
\ttr\left(\left(W_i\right)^{\Delta^i_1}\right)\cdots \ttr\left(\left(W_i\right)^{\Delta^i_{\ell_i}}\right)
\to \ttr\left( W_i\right) .
$$

One may interpret it as a projection of $\ell_i$ polygons to the country (the polygon) labeled by $i$.

\er

\br
Notice that the answers for the expectation values which were considered above 
depend only on eigenvalues of 
$W^*_i$ or $W_i$ $ i=1,\dots,\V$.

\er

\section{Examples of Matrix Models}

Recall that in Section \ref{Partitions-and-Schur-functions} we introduced the function 
$\tau_r(n,\bpow^{(1)},\bpow^{(2)})$.
In what follows we use the~conventions:
\be\label{convention}
\tau_r(\bpow^{(1)},\bpow^{(2)}):=\tau_r(0,\bpow^{(1)},\bpow^{(2)})=
\sum_\lambda r_\lambda s_\lambda(\bpow^{(1)})s_\lambda(\bpow^{(2)}),\quad 
r_\lambda=r_\lambda(0)
\ee

It depends on two sets $\bpow^i=(p^{(i)}_1,p^{(i)}_2,\dots),\, i=1,2$ , and on the choice of 
an arbitrary function of the variable $r$. 
\hl{(This is an example of the so-called tau function, but we will not use this fact.)}
As one of their sets, we will choose $\bpow^2=\bpow(X)$ like in (\ref{tau(nXp)}), and the second set
will be the set of arbitrary parameters. With $r=1$ we get 
\be\label{tau-vac}
\tau_1(\bpow,X)=e^{\sum_{m>0}\frac{1}{m}p_m\tr\left( X^m\right)}
\ee

For example, if we take 
$$
r(x)=\frac{\prod_i^p(a_i+x)}{\prod_i^q(b_i+x)},
$$
and in addition to $\bpow^{1}=(1,0,0,\dots)$
we get the so-called hypergeometric function of the matrix argument:
\be\label{p-F-q}
{_pF}_q\left({a_1,\dots,a_p\atop b_1,\dots,b_q}\arrowvert X\right)=
\sum_\lambda \frac{{\rm dim}\,\lambda}{|\lambda|!}s_\lambda(X)
\frac{\prod_i^p(a_i+x)_\lambda}{\prod_i^q(b_i+x)_\lambda}
\ee 

Special cases: 
\be\label{e^tr}
e^{\tr X}=\sum_\lambda s_\lambda(X)\frac{{\rm dim}\,\lambda}{|\lambda|!},\quad 
\ee
\be\label{det(1-X)}
\det(1-zX)^{-a} = \sum_\lambda z^{|\lambda|} (a)_\lambda s_\lambda(X)\frac{{\rm dim}\,\lambda}{|\lambda|!}
\ee

{Integrals.} Using Proposition \ref{prop-Schur}  we obtain

\begin{Theorem} Suppose $W_1,\dots, W_\f$ and $W_1^*,\dots, W_\V^*$ are dual sets (\ref{duality}).
Let sets $\bpow^i=\left(p^{(i)}_1,p^{(i)}_2,p^{(i)}_3,\dots\right)$, $i=1,\dots,{\rm max}(\f,\V)$ be independent complex
parameters and $r{(i)},\,i=1,\dots,{\rm max}(\f,\V)$ be a set of given functions in one variable.
\be\label{E-tau-X-p}
 \E_{n_1,n_2}\left\{ \Dr_X\left[ \tau_{r^{(1)}}\left(\bpow^1, W_1 \right) \cdots
 \tau_{r^{(\f)}}\left(\bpow^\f, W_\f \right)\right]\right\}  
\ee 
\be \label{Th1}
= \sum_\lambda \,r_\lambda\,
\hbar^{n_1|\lambda|} 
\left(\frac{{\rm dim}\,\lambda}{|\lambda|!}\right)^{-n}
\prod_{i=1}^\V s_{\lambda}\left(W_i^*\right)
\prod_{i=1}^\f s_\lambda(\bpow^i),
\ee
where each $\tau_{r^{(i)}}\left(\bpow^i, W_i \right)$ is defined by (\ref{tau(nXp)})
$$
r_\lambda=\left((N)_\lambda\right)^{-n_2}
\prod_{i=1}^\f r^{(i)}_\lambda(n)
$$

Similarly
\be\label{E-tau-X-p*}
 \E_{n_1,n_2}\left\{ \Dr_X\left[ \tau_{r^{(1)}}\left(\bpow^1, W_1^* \right) \cdots
 \tau_{r^{(\V)}}\left(\bpow^\V, W_\V^* \right)\right]\right\}  
\ee 
\be \label{Th1*}
= \sum_\lambda \,r_\lambda\,
\hbar^{n_1|\lambda|} 
\left(\frac{{\rm dim}\,\lambda}{|\lambda|!}\right)^{-n}
\prod_{i=1}^\f s_{\lambda}\left(W_i\right)
\prod_{i=1}^\V s_\lambda(\bpow^i),
\ee

\end{Theorem}
 
\br We recall the convention (\ref{convention})
 In (\ref{Th1}) $r_\lambda$ is the content product (\ref{content-product})
$$ 
r_\lambda=r_\lambda(0)=\prod_{(i,j)\in\lambda} r(j-i)
$$
 where
$$
r(x)=\left(N+x  \right)^{-n_2}\prod_{i=1}^\f r^{(i)}(x)
$$
\er 

To get examples we choose 

\begin{itemize}
 \item Dual sets $W_1,\dots,W_\f\leftrightarrow W_1^*,\dots,W_\V^*$;
 \item The fraction of unitary matrices given by $n_2$;
 \item The set of functions $r^{(i)},\,i=1,\dots,\f$;
 \item The sets $\bpow^{(i)},\,i=1,\dots,\f$.
\end{itemize}

\br\label{simplifications} 
Answers in some cases are further simplified. Let us mark two
cases

(i) Firstly, this is the case when 
the spectrum of the stars has the form 
\be\label{I-N-k}
{\rm Spect} \,W_i^* = {\rm Spect}\,\mathbb{I}_{N,k_i}=\diag \{1,1,\dots,1,0,0,\dots,0 \},\quad i=1,\dots,\V
\ee
where $\mathbb{I}_{N,k_i}$ is the matrix
with $k_i$ units of the main diagonal. Such star monodromies  obtained in case source matrices have a rank
smaller than $N$.
Insertion of such matrices in the left-hand sides of (\ref{Th1}) and (\ref{Th1*}) corresponds to the integration
over rectangular random matrices.
One should take into account
that
\[
 s_\lambda(\mathbb{I}_{N,k})=(k)_\lambda s_\lambda(\bpow_\infty),
\]
where we recall the notation
\be\label{Pochhammer-lambda}
 (a)_\lambda:=(a)_{\lambda_1}(a-1)_{\lambda_2}\cdots (a-\ell+1)_{\lambda_\ell},
\ee

(ii) The case is the specification of the sets $\bpow^i$, $i=1,\dots,\f$ according to the following

\bl\label{specializations} Denote 
\be\label{p_infty}
\bpow_\infty =(1,0,0,\dots)
\ee
\be\label{p(a)}
\bpow(a)=\left(a,a,a,\dots \right)
\ee
\be\label{p(t,q)}
\bpow(\texttt{q},\texttt{t})=\left(p_1(\texttt{q},\texttt{t}),p_2(\texttt{q},\texttt{t}),\dots\right)\,,\quad
p_m(\texttt{q},\texttt{t})=  \frac{1-\texttt{q}^m}{1-\texttt{t}^m}
\ee 

Then
\be\label{Schur-t(a)}
\frac{s_\lambda(\bpow(a))}{s_\lambda(\bpow_\infty)}=(a)_\lambda\,,\quad \bpow(a)=(a,a,a,\dots)
\ee
where $(a)_\lambda:=(a)_{\lambda_1}(a-1)_{\lambda_2}\cdots (a-\ell+1)_{\lambda_\ell}$, $(a)_n:=a(a+1)\cdots(a+n-1)$,
where $\lambda=(\lambda_1,\dots,\lambda_\ell)$ is a partition.
More generally
\be\label{Schur-t(t,q)}
\frac{s_\lambda(\bpow(\texttt{q},\texttt{t}))}{s_\lambda(\bpow(0,\texttt{t}))}=(\texttt{q};\texttt{t})_\lambda\,,
\ee
where $(\texttt{q};\texttt{t})_\lambda =
(\texttt{q};\texttt{t})_{\lambda_1}(\texttt{q t}^{-1};\texttt{t})_{\lambda_2}\cdots
(\texttt{q t}^{1-\ell};\texttt{t})_{\lambda_\ell}$ where
$(\texttt{q};\texttt{t})_k=(1-\texttt{q})(1-\texttt{q t})\cdots (1-\texttt{q t}^{n-1}) $
is $\texttt{t}$-deformed Pochhammer symbol. $(\texttt{q};\texttt{t})_0=1$ is implied.
\el

For such specifications the right-hand side of (\ref{Th1}) can ta 

With such specifications, one can diminish the number of the Schur functions in the right-hand side of 
(\ref{Th1}) (or of (\ref{Th1*})) and the right-hand side can take one of the forms:
\be\label{form1}
\sum_\lambda r_\lambda s_\lambda(A)s_\lambda(B)
\ee
\be\label{form2}
\sum_\lambda r_\lambda s_\lambda(A)
\ee
\be\label{form3}
\sum_\lambda r_\lambda 
\ee

For (\ref{form1}) there is a determinant representation; for (\ref{form2}) there is a Pfaffian representation
and (\ref{form3}) can be rewritten as a sum of products.
Indeed if we introduce $r(x)=e^{T_{x-1}-T_x}$, and $\lambda=(\lambda_1,\dots,\lambda_N)$, then
\be
r_\lambda(m)=\prod_{(i,j)\in\lambda}r(m+j-i)=e^{T_m+\cdots+T_{m-N}}\prod_{i=1}^N e^{T_{\lambda_i-i+m}}
\ee

For instance, one can take $r(x)=a+x$ and get
\be\label{(a)_lambda-poch)}
(a)_\lambda =\frac{\Gamma(a+\lambda_1-1)\Gamma(a+\lambda_2-2)\cdots \Gamma(a+\lambda_N-N)}
{\Gamma(a)\Gamma(a-1)\cdots \Gamma(a-N+1)}
\ee

Then we introduce $h_i=\lambda_i-i+N$ and write
\be\label{pochh-via-gamma}
\sum_\lambda\frac{(a)_\lambda}{(b)_\lambda}=
\sum_{h_1>\cdots >h_N\ge 0}\prod_{i=1}^N \frac{\Gamma(b-i+1)}{\Gamma(a-i+1)}\frac{\Gamma(h_i+a-N)}{\Gamma(h_i+b-N)}
\ee
where $\Gamma$ is the gamma-function. 
(In case the argument of gamma-function turns out to be a nonpositive integer one should keep in mind 
both the enumerator and denominator.)
 See examples below.

\er

\bx\label{E1} See Example \ref{1} and Figure \ref{fig:figure2}a. Take $X_1=Z$ and $r$ given by(\ref{r-rational}). 
The example of (\ref{Th1}) can be chosen as follows
\be
\E_{1,0}\left\{  
{_pF}_q\left({a_1,\dots,a_p\atop b_1,\dots,b_q}\arrowvert ZC_1\right)
{_{p'}F}_{q'}\left({a'_1,\dots,a'_{p'}\atop  b'_1,\dots,b'_{q'}}\arrowvert Z^\dag C_{-1}\right)
\det\left(Z Z^\dag  \right)^{\alpha}
\right\}=
\ee
\be
{_{p'}F}_{q'}\left({a_1,\dots,a_p,a'_1,\dots,a'_{p'},N+\alpha \atop b_1,\dots,b_q,
b'_1,\dots,b'_{q'},N}\arrowvert C_1 C_{-1}\right),
\ee
corresponding determinantal representation see a (\ref{tau(nXp)}).

See  and Figure \ref{fig:figure2}b which is dual to Figure \ref{fig:figure2}a.
An example of (\ref{Th1*}) can be chosen as
$$
\E_{1,0}\left\{
{_{p'}F}_{q'}\left({a_1,\dots,a_p,a'_1,\dots,a'_{p'},N+\alpha \atop b_1,\dots,b_q,
b'_1,\dots,b'_{q'},N}\arrowvert ZC_1 Z^\dag C_{-1}\right)\det\left(Z Z^\dag  \right)^{\beta}
\right\}=
$$
\be\label{}
\sum_\lambda s_\lambda(C_1)s_\lambda(C_{-1})
\frac{(N+\alpha)_\lambda(N+\beta)_\lambda)}{\left( (N)_\lambda \right)^2}
\frac{\prod_i^p(a_i)_\lambda}{\prod_i^q(b_i)_\lambda}\frac{\prod_i^{p'}(a_i)_\lambda}{\prod_i^{q'}(b_i)_\lambda},
\ee

The determinantal representation of the left-hand side is given by (\ref{tau(nXY)}).
\ex

\bx\label{E2} See Example \ref{2} and Figure \ref{fig:figure5}a. Take $X_1=U_1,\,X_2 =U_2$.
\be
\E_{0,2}\left\{e^{\sum_{m>0} \frac 1m p_m \tr\left(U_1C_1U_2C_2U_1^\dag C_{-1}U_2^\dag C_{-2}\right)^m}
\right\}=
\ee
$$
 =\sum_{\lambda} 
 \frac{1}{\left((N)_\lambda\right)^2}
 \frac{s_\lambda(\bpow) s_\lambda(C_1C_{-2}C_{-1}C_2)}{\left(s_\lambda(\bpow_\infty)\right)^2}
$$

Let us take $\bpow=(az,az^2,az^3,\dots)$ and $W_1^*=\mathbb{I}_{N,k}$; see (\ref{p(a)}) and (\ref{I-N-k}).
We obtain the left-hand side as
$$
\E_{0,2}\left\{\det\left(1-zU_1C_1U_2C_2U_1^\dag C_{-1}U_2^\dag C_{-2}\right)^{-a}
\right\}=\sum_\lambda z^{|\lambda|}\frac{(a)_\lambda (k)_\lambda}{\left((N)_\lambda\right)^2}=
$$
$$
z^{-\tfrac 12 N(N-1)}
\sum_{h_1>\cdots >h_N\ge 0}\prod_{i=1}^Nz^{h_i}\frac{\left(\Gamma(N-i+1)\right)^2}{\Gamma(a-i+1)\Gamma(k-i+1)}
\frac{\Gamma(h_i+a-N)\Gamma(h_i+k-N)}{\left(\Gamma(h_i)\right)^2}
$$
where ${\rm Spect}\,C_1C_{-2}C_{-1}C_2={\rm Spect}\,\mathbb{I}_{N,k}$
see (\ref{pochh-via-gamma}).
\ex

\bx\label{E3} See Example \ref{3} and Figure \ref{fig:figure5}c.
$$
\E_{2,0}\left\{e^{\sum_{m>0} \frac 1m p^{(1)}_m \tr\left(Z_1C_1)^m\right)+
\sum_{m>0} \frac 1m p^{(2)}_m \tr\left(Z_2C_2)^m\right)}
\det\left(1-z Z_1^\dag C_{-1}Z_2^\dag C_{-2}\right)^{-a}\prod_{i=1}^3
\det\left(Z_i Z_i^\dag  \right)^{\alpha}
\right\}
$$
$$
=\sum_\lambda z^{|\lambda|} s_\lambda(\bpow^1)s_\lambda(\bpow^2) (a)_\lambda
\left(\frac{ (N+\alpha)_\lambda}{(N)_\lambda}\right)^3
$$

For a determinant representation see (\ref{tau(npp)}) 

\ex

\bx\label{E4} For decoration of  Figure \ref{fig:figure1}e  we put $X_i=Z_i$.
$$
\E_{2,0}\left\{\bpow_\Delta(Z_1^\dag C_{-1})s_\lambda(Z_1C_1Z_2C_2Z_2^\dag C_{-2}) 
\right\} = \delta_{|\lambda|,|\Delta|}\frac{{\rm dim}\,\lambda}{|\lambda|!} \chi_\lambda(\Delta)
s_\lambda(C_2)s_\lambda(C_1C_{-2}C_{-1})
$$
 see Proposition \ref{MM1-all-skew}.
\ex 

\bx\label{E5} Figure \ref{fig:figure4}d in particular yields
$$
\E_{2,0}\left\{
s_\lambda(Z_1C_1Z_2C_2Z_3C_3Z_4C_4)s_\lambda(Z_3^\dag C_{-3} Z_2^\dag C_{-2}Z_5C_5)
s_\lambda(Z_5^\dag C_{-5}Z_1^\dag C_{-1}Z_4^\dag C_{-4})\right\} = 
$$
$$
\left(\frac{{\rm dim}\,\lambda}{|\lambda|!}\right)^{-5}
s_\lambda(C_1C_{-2}C_{-5}) s_\lambda(C_{2}C_{-3}) s_\lambda(C_5C_3C_{-4}) s_\lambda(C_4C_{-1})
$$.
\ex

\bx\label{E6} In the case below we use an open chain with $n$ edges as in Figures \ref{fig:figure2}b, 
\ref{fig:figure1}b and \ref{fig:figure3}a.
\begin{myequation}
\E_{n,0}\left\{
e^{\sum_{m>0}\frac 1m p^{(1)}_m \tr\left((Z_1C_{1}Z_2C_2\cdots Z_nC_{n}Z_n^\dag C_{-n} 
\cdots Z_1^\dag C_{-1})^m\right)} \right\}=
\sum_{\lambda} s_\lambda(C_n) s_\lambda(C_{-1}) \frac{s_\lambda(\bpow)}{s_\lambda(\bpow_\infty)}
\prod_{i=1}^{n-1} \frac{s_\lambda(C_{i}C_{-i-1})}{s_\lambda(\bpow_\infty)}
\end{myequation}

Graphs dual to the chain look like in Figure \ref{fig:figure3}b.
$$
\E_{n,0}\left\{e^{\sum_{m>0}\frac 1m p_m \tr\left((Z_1^\dag C_{-1})^m\right)+
\sum_{m>0}\frac 1m p^{(n)}_m \tr\left((Z_n C_{n})^m\right)}
\prod_{i=1}^{n-1}
e^{\sum_{m>0}\frac 1m p^{(i)}_m \tr\left((Z_iC_{i}Z_{i+1}^\dag C_{-i-1})^m\right)}
\right\}
$$
$$
=\sum_{\lambda} s_\lambda(\bpow)s_\lambda(C_{1}C_2\cdots C_{n}C_{-n} \cdots C_{-1})
\prod_{i=1}^{n}\frac{s_\lambda(\bpow^{i})}{s_\lambda(\bpow_\infty)}
$$

These relations generalize product (\ref{zz-dag}) and (\ref{z,z-dag})).
\ex

\bx\label{E7} Our graph is a polygon with $n$ edges and $n$ vertices (stars); see Figures \ref{fig:figure2}a,
\ref{fig:figure1}d and \ref{fig:figure3}c for examples.
\be
\E_{n,0}\left\{
e^{\sum_{m>0}\frac 1m p^{(1)}_m 
\tr\left( Z_1C_{1}Z_2C_2\cdots Z_nC_n  \right)^m
+
\sum_{m>0}\frac 1m p^{(2)}_m \tr\left(Z_n^\dag
C_{-n}Z_{n-1}^\dag C_{1-n}\cdots Z_1^\dag C_{-1}\right)^m }
\prod_{i=1}^n
\det\left(Z_i Z_i^\dag  \right)^{\alpha}
\right\}  
\ee
$$
 =\sum_\lambda s_\lambda(\bpow^1)s_\lambda(\bpow^2)
 \prod_{i=1}^n\frac{(N+\alpha)_\lambda s_\lambda(C_{i}C_{-i-1})}{(N)_\lambda s_\lambda(\bpow_\infty)} 
$$
(where we put $C_{-n-1}=C_{-1}$).

A graph dual to the polygon can be viewed as two-stars graph with $n$ edges which connect stars; see~Figures~\ref{fig:figure1}d and \ref{fig:figure3}d as examples.
\be
\E_{n,0}\left\{\prod_{i=1}^n e^{\sum_{m>0}\frac 1m p_m^{(i)} 
\tr\left( Z_i C_{i}Z_{i+1}^\dag C_{-i-1}\right)^m}
  \right\} =  \sum_\lambda 
  s_\lambda(C_{1}C_2\cdots C_n) s_\lambda(C_{-n}C_{1-n}\cdots C_{-1})
  \prod_{i=1}^n \frac{s_\lambda(\bpow^i)}{ s_\lambda(\bpow_\infty)} 
\ee

To apply determinantal formulas one should use Remark \ref{simplifications}.

\ex

\bx\label{E8} Consider the star-graph with $n$-rays which end at other stars (see Figure \ref{fig:figure4}a where $n=3$). This~situation corresponds to (\ref{zzdag-zzdag}).

\be
\E_{n,0}\left\{e^{\sum_{m>0}\frac 1m p_m 
\tr\left( Z_1C_{1}Z_1^\dag C_{-1}\cdots Z_nC_{n}Z_{n}^\dag C_{-n}\right)}\prod_{i=1}^n
\det\left(Z_i Z_i^\dag  \right)^{\alpha}
\right\}= 
\ee
$$
\sum_\lambda s_\lambda(\bpow)s_\lambda(C_{-1}C_{-2}\cdots C_{-n})
\prod_{i=1}^n \frac{(N+\alpha)_\lambda s_\lambda(C_i) }{(N)_\lambda s_\lambda(\bpow_\infty)}
$$

A similar model was studied in \cite{Chekhov-2014,ChekhovAmbjorn}.
It has the determinantal representation (\ref{tau(nXp)}) in case all $W_i^*$ except one are of form
(\ref{I-N-k}). There is the determinantal representation (\ref{tau(nXY)}) in case we specialize
the set $\bpow$ according to Lemma~\ref{specializations} and choose each $W_i^*$ except two
be in form (\ref{I-N-k}).

Now, let us choose the dual graph (this is petel graph. (see Figure \ref{fig:figure4}b where $n=3$))
and consider
\be
\E_{n,0}\left\{\det\left(1-z Z_1^\dag C_{-1}Z_2^\dag C_{-2}\cdots Z_n^\dag C_{-n}\right)^{-a}
\prod_{i=1}^n \det\left(Z_i Z_i^\dag  \right)^{\alpha} 
e^{\sum_{m>0}\frac 1m p_m^{(i)} \tr\left((Z_iC_i)^m \right) }
\right\}=
\ee
$$
\sum_\lambda z^{|\lambda|}(a)_\lambda s_\lambda(C_1C_{-1}C_2C_{-2}\cdots C_nC_{-n})
\prod_{i=1}^n \frac{(N+\alpha)_\lambda s_\lambda(\bpow^i)}{(N)_\lambda s_\lambda(\bpow_\infty)}
$$

By Remark \ref{simplifications} we find all cases where the determinantal representations 
(\ref{tau(npp)}) or (\ref{tau(nXp)}) exist.

\br\label{powers-symmetry} Notice the following symmetry: the left-hand side produces the same
right-hand side if we permute the set of exponents $\alpha_1,\dots,\alpha_n,a-N$.

\er

\ex

\bx\label{E9} Below $g=1,2,\dots.$ (For the case $g=1$ see Figure \ref{fig:figure1}a and 
zoomed Figure \ref{fig:figure5}a; for $g=2$ see
Figure~\ref{fig:figure4}e).
$$
\E_{0,2g}\left\{e^{\sum_{m>0}\frac 1m p_m \tr\left(
U_{a_1}C_{a_1}U_{b_1}C_{b_1}U_{a_1}^\dag C_{-a_1} U_{b_1}^\dag C_{-b_1}\cdots
U_{a_g}C_{a_g} U_{b_g}C_{b_g}U_{a_g}^\dag C_{-a_g}U_{b_g}^\dag C_{-b_g}\right)^m}\right\}
$$
$$
=\sum_{\lambda} \left(\frac{{\rm dim}\,\lambda}{|\lambda|!}\right)^{-2g}
s_\lambda(W^*)s_\lambda(\bpow)
$$
where
$$
W^*=C_{-a_1}C_{b_1}C_{a_1}C_{-b_1}\cdots  C_{-a_g}C_{b_g}C_{a_g}C_{-b_g}
$$

In particular, if ${\rm Spect}\,W^*=\mathbb{I}_{N,k}$, then
$$
 \E_{0,2g}\left\{\det\left(\mathbb{I}_N -zU_{a_1}C_{a_1}U_{b_1}C_{b_1}U_{a_1}^\dag C_{-a_1} U_{b_1}^\dag C_{-b_1}\cdots
U_{a_g}C_{a_g} U_{b_g}C_{b_g}U_{a_g}^\dag C_{-a_g}U_{b_g}^\dag C_{-b_g}\right)^{-a}\right\}
$$
\be\label{beta-ensemble}
=\sum_{\lambda} z^{|\lambda|} \left(\frac{{\rm dim}\,\lambda}{|\lambda|!}\right)^{2-2g}
\frac{(a)_\lambda (k)_\lambda}{\left((N)_\lambda  \right)^{2g}}
\ee

Taking into account that
$$
{\rm dim}\,\lambda=\frac{\prod_{i<j}(h_i-h_j)}{\prod_{i=1}^N\Gamma(h_i+1)}
$$
we can interpret that (\ref{beta-ensemble}) is a discrete beta-ensemble where $\beta=2-2g$.
\ex

{Exotic models. An example.} There are some more tricky problems which can be solved which can be solved in steps.
Let me consider the simplest example. Look at Lemma \ref{Lemma-s-s-Z}.
Suppose $A$~and~$B$ depend in any way on an additional matrix $Z_1$;
in any case, however, their product has a familiar form:
\be 
 A=A(Z_1,Z^\dag_1),\quad B=B(Z_1,Z^\dag_1),\quad AB=Z_1C_1Z^\dag_1C_{-1}
\ee

Say, $A=Z_1^{a}e^{-Z_1^\dag},\,B=e^{Z_1^\dag}Z_1^\dag Z_1^{1-a} $ which looks horrible.
However, applying sequentially the series ${_0F}_0$, then  (\ref{s(ZA)s(ZB)}),
where $Z=Z_1$, and
then (\ref{s(ZAZB)}), where $Z=Z_2$  one obtains
\be
\E_{2,0} 
\left\{ 
e^{\ttr\left(Z_1Z_2^{a}e^{-Z_2^\dag}\right) +
\ttr\left(Z^\dag_1e^{Z_2^\dag}Z_2^\dag Z_2^{1-a} \right)} 
\right\}=\sum_{d\ge 0}\sum_{\lambda\in\Upsilon_d} s_\lambda(\mathbb{I}_N)s_\lambda(C)
=\det\left(\mathbb{I}_N - C\right)^{-1}
\ee

It will be interesting to do the same with other ensembles of random matrices; Ginibre ensembles of
real and quaternionic matrices; and ensembles of Hermitian matrices: complex, real and quaternionic.

\section{Discussion}

In this article, we examined matrix integrals of a certain type. We called them matrix models associated with 
children's drawings---the so-called dessin d'enfants. They include some well-known models that have found 
applications in the theory of information transfer and the theory of quantum~chaos. We hope that our matrix 
integrals will be in 
demand. We think that these models are related to quantum integrable systems \cite{Migdal,Rusakov,Witten},
but this topic is waiting for 
its development; we expect { connections with } 
\cite{Olshanski-19,Olshanski-199,Okounkov-1,Okounkov-19,Okounkov-199,Okounkov-1996,VershikOkounkov,Gerasimov-Shatashvili,Rumanov}.

\vspace{6pt}

\section{Acknowledgments}
The authors are grateful to A.Gerasimov, M.Kazarian, S.Lando, Yu.Neretin, A.Morozov, A.Mironov, S.Natanzon and L.Chekhov for useful discussions.
A.O. is grateful to A.Odzijewicz for his kind hospitality in Bialowezie and to  
E.Strahov, who turned his attention to independent Ginibre ensembles \cite{Ak1,S2,S1}.
A.O. was partially supported by V.E. Zakharov's scientific school
(Program for Support of Leading Scientific Schools), \mbox{by RFBR} grant 18-01-00273a.
N.A. was partially supported by RFBR grant 19-02-00815. D.V. was partially supported by RFBR grant 18-02-01081.

\appendix
\section{Partitions and Schur Functions \label{Partitions-and-Schur-functions-}}

\hl{Let us recall that the characters of the unitary group} $\mathbb{U}(N)$ are labeled by partitions
and coincide with the so-called Schur functions \cite{Mac}.   
A partition 
$\lambda=(\lambda_1,\dots,\lambda_n)$ is a set of nonnegative integers $\lambda_i$ which are called
parts of $\lambda$ and which are ordered as $\lambda_i \ge \lambda_{i+1}$. 
The number of non-vanishing parts of $\lambda$ is called the length of the partition $\lambda$, and will be denoted by
 $\ell(\lambda)$. The number $|\lambda|=\sum_i \lambda_i$ is called the weight of $\lambda$. The set of all
 partitions will be denoted by $\mathbb{P}$.

The Schur function labelled by $\lambda$ may be defined as  the following function in variables
$x=(x_1,\dots,x_N)$ :
\be\label{Schur-x-}
 s_\lambda(x)=\frac{\det \left[x_j^{\lambda_i-i+N}\right]_{i,j}}{\det \left[x_j^{-i+N}\right]_{i,j}}
 \ee
 in case $\ell(\lambda)\le N$ and vanishes otherwise. One can see that $s_\lambda(x)$ is a symmetric homogeneous 
 polynomial of degree $|\lambda|$ in the variables $x_1,\dots,x_N$, and $\deg x_i=1,\,i=1,\dots,N$.
  
 \br\label{notation} In case the set $x$ is the set of eigenvalues of a matrix $X$, we also write $s_\lambda(X)$ instead
 of $s_\lambda(x)$.
 \er

 There is a different definition of the Schur function as a quasi-homogeneous, non-symmetric polynomial of degree $|\lambda|$ in 
 other variables, the so-called power sums,
 $\bpow =(p_1,p_2,\dots)$, where~$\deg p_m = m$.
 
For this purpose let us introduce 
$$
 s_{\{h\}}(\mathbf p)=\det[s_{(h_i+j-N)}(\mathbf p)]_{i,j},
$$
where $\{h\}$ is any set of $N$ integers, and where
the Schur functions $s_{(i)}$ are defined by $e^{\sum_{m>0}\frac 1m p_m z^m}=\sum_{m\ge 0} s_{(i)}(\bpow) z^i$.
If we put $h_i=\lambda_i-i+N$, where $N$
is not less than the length of the partition $\lambda$; then
\begin{equation}\label{Schur-t}
 s_\lambda(\mathbf p)= s_{\{h\}}(\mathbf p).
\end{equation}

 The Schur functions defined by (\ref{Schur-x}) and by (\ref{Schur-t}) are equal,  $s_\lambda(\bpow)=s_\lambda(x)$, 
 provided the variables $\bpow$ and $x$ are related by the power sums relation
  \be
\label{t_m}
  p_m=  \sum_i x_i^m
  \ee
  
  In case the argument of $s_\lambda$ is written as a non-capital fat letter  the definition (\ref{Schur-t}),
  and we imply the definition (\ref{Schur-x}) in case the argument is not fat and non-capital letter, and
  in case the argument is capital letter which denotes a matrix, then it implies the definition (\ref{Schur-x}) with $x=(x_1,\dots,x_N)$ being
  the eigenvalues.
  
  It may be easily checked that
  \be\label{p-to-p-in-Schur}
  s_\lambda(\bpow)=(-1)^{|\lambda|}s_{\lambda^{\rm tr}}(-\bpow)
  \ee
  where $\lambda^{\rm tr}$ is the partition conjugated to $\lambda$ (in \cite{Mac} it is denoted by $\lambda^*$). The Young diagram
  of the conjugated partition is obtained by the transposition of the Young diagram of $\lambda$ with respect to its main diagonal. 
  One gets $\lambda_1=\ell(\lambda^{\rm tr})$.

  \section{Integrals over the Unitary Group\label{IOUG}}
 Consider the following integral over the unitary group which depends on two semi-infinite sets of parameters
 $\bpow=(p_1,p_2,\dots )$ and $\bbpow=(p_1^*,p_2^*,\dots ) $: 
  \be
I_{\mathbb{U}(N)}(\bpow,\bbpow):= \int_{\mathbb{U}(N)} 
e^{\tr V\left(\bpow ,U\right) + \tr V\left(\bpow^*,U^{-1}\right)}  d_*U=
 \ee
 \be
\frac{1}{(2\pi )^N} 
\int_{0 \le \theta_1 \le  \dots \le \theta_N\le 2\pi} 
\prod_{1\le j<k\le N}\vert e^{i\theta_j}-e^{-i\theta_k} \vert ^2 
 \prod_{j=1}^N e^{\sum_{m>0}\frac 1m \left(p_me^{im\theta_j} +p_m^* e^{-im\theta_j}\right)}d\theta_j
 \ee
 \be\label{V}
 V(\bpow,x):= \sum_{n>0} \frac 1n p_n x^n
 \ee  
 
 Here $d_*U$ is the Haar measure of the group $\mathbb{U}(N)$:
\be\label{Haar-unitary}
 d_*U =\frac{1}{(2\pi )^N}  
\prod_{1\le j<k\le N}\vert e^{i\theta_j}-e^{-i\theta_k} \vert ^2 
 \prod_{j=1}^N d\theta_j\,,\quad -\pi \le \theta_1<\dots\theta_N\le \pi
 \ee
 and
 $e^{i\theta_1},\dots,e^{i\theta_N }$ are the eigenvalues of $U\in  \mathbb{U}(N)$. The exponential factors
 inside the integral may be treated as a perturbation of the Haar measure and parameters $\bpow, \, \bpow^*$
 are called coupling constants by the analogy with quantum field theory problems.

 Using  the Cauchy-Littlewood identity
 \be
\label{CL}
  \tau(\bpow|\bpow^*):=e^{\sum_{m=1}^\infty \frac 1m p_m^*p_m}=\sum_{\lambda\in \mathbb{P}} s_\lambda(\bpow^*)s_\lambda(\bpow)
  \ee
 and the orthogonality of the irreducible characters of the unitary group
  \be\label{orthonormality-ch-U}
  \int s_\lambda(U)s_\mu(U^{-1})d_*U = \delta_{\lambda,\mu}
  \ee
 we obtain that
 \be\label{Morozov}
 I_{\mathbb{U}(n)}(\bpow,\bbpow) = \sum_{\lambda\in\mathbb{P}\atop
 \ell(\lambda)\le n} s_\lambda(\bpow) s_\lambda(\bbpow)
 \ee
which express the integral over unitary matrices as the "perturbation series in coupling constants."

The formula (\ref{Morozov}) first appeared in \cite{MirMorSem} in the context of the study of Brezin--Gross--Witten model.
It was shown there that the integral $ I_{\mathbb{U}(n)}(\bpow,\bbpow)$ may be related to 
the Toda lattice tau function
of \cite{JM,UT} under certain restriction. Then, the series in the Schur functions (\ref{Morozov}) may be related 
to the double Schur functions series found in \cite{Takasaki,Takebe}.

\section{Geometrical Definition of Hurwitz Numbers \label{Hurwitz-geometric-section}}

In this presentation, we follow article \cite{NO2020}.

The Hurwitz number is a characterisation of the branched covering of a surface with critical values of a 
prescribed topological type.  
Hurwitz numbers of oriented surfaces without boundaries were introduced by Hurwitz
at the end of the 19th century.
Later it turned out that they are closely related to the study of moduli spaces of Riemann
surfaces \cite{ELSV}, to  integrable systems \cite{Okounkov-2000}, to modern models of mathematical physics
(matrix models) and to closed topological field theories \cite{Dijkgraaf}. In this paper we  consider 
only Hurwitz numbers of compact surfaces without boundary. 

 Consider a branched covering $f:P\rightarrow\Sigma$ of degree $d$ of a compact surface
without boundary. In the neighborhood of each point $z\in P$, the map $f$ is topologically equivalent to the
complex map $u\mapsto u^p$, defined on a neighborhood  $u\sim0$ in $\mathbb{C}$. The number 
$p=p(z)$ is called the degree of the
covering $f$ at the point $z$. The point $z\in P$ is said to be a \textit{branch point} or  
\textit{critical point} if $p(z)\neq 1$. There~are only a finite number of critical points. 
The image $f(z)$ of a critical point $z$ is 
called the \textit{critical value} of $f$ at $z$.

Let us associate with a point $s\in\Sigma$ all points  $z_1,\dots,z_\ell\in P$ for which $f(z_i)=s$. Let
$p_1, \dots,p_\ell$ be the degrees of the map $f$ at these points. Their sum $d=p_1 +\dots+p_\ell$ is equal
to the degree $d$ of $f$. Thus, to each point $s\in S $ there corresponds a partition $d=p_1 +\dots+p_\ell$ of
the number $d$. 
Having ordered the degrees $ p_1 \geq \dots \geq p_\ell> 0 $ at each point $ s \in \Sigma $, 
we  introduce the Young diagram $ \Delta^s = [p_1, \dots, p_ \ell ] $ of weight $ d $ with
$ \ell = \ell (\Delta^s) $ rows of length $ p_1 \dots, p_\ell $ :
 $ \Delta^s $ is called the 
\textit{topological type} of the value $ s $, and 
 $ s $ is a critical value of $f$ if and only if at least one of the row-lengths $ p_i $
 is greater than $ 1 $.)

Let us note that the Euler characteristics $\e(P)$ and $\e(\Sigma)$ of the surfaces $P$ and $\Sigma$ are related
via the Riemann--Hurwitz relation:
\[
\e(P)=\e(\Sigma)d +\sum\limits_{z\in P} \left(p(z)-1\right)
\]
or, equivalently,
\be\label{RHur}
\e(P)=\e(\Sigma)d +\sum\limits_{i=1}^\f \left( \ell(\Delta^{s_i})-d\right).
\ee
where $s_1,\dots,s_{\f}$ are critical values.

We say that coverings $ f_1: P_1 \rightarrow \Sigma $ and $ f_2: P_2 \rightarrow \Sigma $ 
are \textit{equivalent}
if there exists a homeomorphism $ F: P_1 \rightarrow P_2 $ such that $ f_1 = f_2F $;
in case $P_1=P_2$ and $f_1 =f_2$ the homeomorphism $F$ is called an
\textit{automorphism of the covering}. 
The set of all automorphisms of a covering $f$ form the group $\texttt{Aut}(f)$
of finite order $|\texttt{Aut}(f)|$. Equivalent coverings have isomorphic automorphism groups.

We present two illustrative examples.

Example 1. Let $ \Sigma = \overline {\mathbb{C}} = \{z \in \mathbb{C} \} \bigcup \infty $,
$ P = P_1 = P_2 = \overline {\mathbb{C}} = \{u \in \mathbb {C} \} \bigcup \infty $  be Riemann spheres.
Consider the branched covering $ z(u) = f (u) = f_1 (u) = f_2 (u) = u^3 $.
This covering $ f: P \to \Sigma $ has 2 critical values 0 and $ \infty $
with Young diagrams from one row of length 3.
Automorphisms of the covering have the form $ F(u) = u^{\sqrt [3] {1}} $.
The group $ \texttt{Aut} (f) $ 
is isomorphic to $ \mathbb{Z} / 3 \mathbb{Z} $.

Example 2. Let $ \Sigma = \overline{\mathbb {C}} = \{z \in \mathbb{C} \} \bigcup \infty $ and
$ P = P_1 = P_2 $ - this is a pair of Riemann spheres; 
that is $ P=P' \bigcup P$."
where $ P' = \{ u'
\in \mathbb{C} \} \bigcup \infty $ and 
$P'' =\{ u''
 \in \mathbb{C} \} 
\bigcup \infty $.
Consider the branched covering $ z(u') = f (u') = f_1 (u') = f_2 (u') = (u')^3 $,
$ z(u'') = f(u'') = f_1 (u'') = f_2 (u'') = (u'')^3 $.
This~covering $ f: P \to \Sigma $ has two critical values 0 and $ \infty $
with Young diagrams of two rows of length 3.
Automorphisms of the covering are generated by the following mappings:

    1. $F(u')=(u')^{\sqrt[3]{1}}$, $F(u'')=u''$.
    
    2. $F(u'')=(u'')^{\sqrt[3]{1}}$, $F(u')=u'$.
    
    3. $F(u')=(u'')$, $F(u'')=u'$.
    
    The group $ \texttt{Aut}(f_i) $ is isomorphic to 
$(\mathbb{Z}/3\mathbb{Z})\bigotimes(\mathbb{Z}/3\mathbb{Z})\bigotimes(\mathbb{Z}/2\mathbb{Z})$.

From now on, unless indicated otherwise, we will assume that the surface $\Sigma$ is connected. Let~us choose 
points $ s_1, \dots, s_\f \in \Sigma $ and corresponding Young diagrams $ \Delta^1, \dots, \Delta^\f $ of 
weight $ d $. 
Let $ \Phi $ be the set of  equivalence classes of the coverings for which $ s_1, \dots, s_\f $ is the set 
of all critical values, and~$ \Delta^1, \dots, \Delta^\f $ are the topological types of these critical values.
The \textit{Hurwitz number} 
is the number
\be\label{disconH} 
H_{\Sigma}^d(\Delta^1,\dots,\Delta^\f)=\sum_{f\in\Phi} 
\frac {1} {|\texttt{Aut} (f)|}.
\ee

It is easy to prove that the Hurwitz number is independent of the positions of the points $s_1, \dots, s_\f$ 
on $\Sigma$.  One can show that the right-hand side of (\ref{disconH})
 depends only on the Young diagrams of $\Delta^1,\dots,\Delta^\f$ and the Euler characteristic $\e=\e(\Sigma)$.
 Because of this sometimes we write $H_{\e(\Sigma)}^d(\Delta^1,\dots,\Delta^\f)$ instead of 
 $H_{\Sigma}^d(\Delta^1,\dots,\Delta^\f)$.

\vspace{1ex}

If $\f=0$ we get an unbranched covering. We denote such Hurwitz number $H_\e\left((1^d)\right)$.

Example 3.
Let $f:\Sigma\rightarrow\mathbb{RP}^2$ be a covering without critical points.
Then, if $\Sigma$ is connected, then $\Sigma=\mathbb{RP}^2$,
$\deg f=1$\quad or $\Sigma=S^2$, $\deg f=2$. Therefore if $d=3$, then
$\Sigma=\mathbb{RP}^2\coprod\mathbb{RP}^2\coprod\mathbb{RP}^2$ or $\Sigma=\mathbb{RP}^2\coprod S^2$.
Thus $H_{1}\left((1^3)\right)=\frac{1}{3!}+\frac{1}{2!}=\frac{2}{3}$.

\section{Combinatorial Definition of Hurwitz Numbers \label{Hurwitz-combinatorial-section}}

Consider the symmetric group (equivalently, the permutation group) $S_d$ and the equation
\be\label{Hurwitz-combinatorial}
\sigma_1\cdots \sigma_\f \rho_1^2\cdots \rho_\textsc{m}^2\alpha_1\beta_1\alpha_1^{-1}\beta_1^{-1}\cdots
\alpha_{\textsc{h}}\beta_{\textsc{H}}\alpha_{\textsc{h}}^{-1}\beta_{\textsc{h}}^{-1}=1,
\ee
where $\sigma_1,\cdots , \sigma_\f, \rho_1,\cdots, \rho_\textsc{m},\alpha_1,\beta_1,
\dots,\alpha_{\textsc{H}},\beta_{\textsc{h}} \in S_d$, and moreover $\sigma_i \in C_{\Delta^i},\, i=1,\dots,\f$, 
where $C_{\Delta^i}$ is the conjugacy class labeled by a partition $\Delta^i=\left(\Delta^i_1,\Delta^i_2,\dots \right)$.
The Hurwitz number is the number of solutions of Equation (\ref{Hurwitz-combinatorial}) divided by $d!$ (by the order
of $S_d$). 

It can be proved that so introduced
the (combinatorial) Hurwitz number coincides with the (geometric) Hurwitz number 
$H_\e(\Delta^1,\dots,\Delta^\f)$ introduced in
Appendix \ref{Hurwitz-geometric-section}   where  $\e=2-2\textsc{h}-\textsc{m}$.
(One~can look at the base surface $\Sigma$ as a result of gluing $\textsc{h}$ handles and $\textsc{m}$ 
M\"obius stripes to a sphere.

Consider the simplest example: $\textsc{h}=0$ and $\textsc{m}=1$; that is $\Sigma=\mathbb{RP}^2$
(real projective plane). Suppose~$\f=0$; that is we deal with an unbranched covering. Suppose
$d=3$; that is we consider 3-sheeted covering. Let us solve $\rho^2=1$, where $\rho\in S_3$. One gets
4 solutions: 3 transpositions of the set $1,2,3$ and one identity permutation. There are $3!$
permutations in $S_3$. As a result we get $H_1\left((1^3)\right)=4/3!=2/3$ as we got in the last example
of the previous section.

In the same way one can consider Example 1 of the previous section. In this case $\textsc{h}=\textsc{m}=0$;
that is $\e=2$; one gets the Riemann sphere with two branch points ($\f=2$) and 3-sheeted covering
with profiles 
$\Delta^1=\Delta^2=(3)$. We solve the equation
$\sigma_1\sigma_2=1$, where both $\sigma_{1,2}$ consist of a single cycle of length 3. There are two
solutions $\sigma_1=\sigma_2^{-1}$: one sends $1,2,3$ to $3,1,2$, the other sends $1,2,3$ to $2,3,1$.
We get $H_2\left((3),(3)\right)=2/3!=1/3$.

Example 2 corresponds to $H_2\left((3,3),(3,3) \right)$, $d=6$. One can complete the exercise and get \mbox{an answer} 
$H_2\left((3,3)\right)=1/z_{(3,3)}=1/18$, where $z_\lambda$ is given by (\ref{z-lambda}). Actually, for any $d$ 
and for any pair of profiles one gets $H_2\left(\Delta^1,\Delta \right)=\delta_{\Delta^1,\Delta}1/z_\Delta$.

In \cite{M1,M2} (and also in \cite{GARETH.A.JONES}) it was found that $H_\e(\Delta^1,\dots,\Delta^\f)$
is given by formula (\ref{Mednyh}).

\section{Differential Operators\label{DiffOp}   }

In \cite{NO2020F}
we develop (see also \cite{NO2020F}) the work \cite{MM1},
which offers a beautiful generalization of the cut-and-join formula (MMN formula):
\be\label{MM3'}
{\cal W}^{\Delta}(\bpow)\cdot s_\mu(\bpow)=\varphi_\mu(\Delta)s_\mu(\bpow),
\ee
which describes the merging of pairs of branch points in the covering problem.
Here $\mu=(\mu_1,\mu_2,\dots)$ and $\Delta=(\Delta_1,\dots,\Delta_\ell )$  are Young diagrams 
($\Delta$ is  the ramification profile of one of branch points;
for~simplicity, we consider the case where $|\mu|=|\Delta|$),
$s_\mu$ is the Schur function.
Differential operators $ {\cal W}^{\Delta} (\bpow) $ generalize the operators of "additional symmetries"
\cite{O-1987} in the theory of solitons and commute with each other for different $ \Delta $.
In the work (\cite{MM3}), it was noted that if they are written in the so-called Miwa variables, that is,
in terms of the eigenvalues of the matrix $X$ such that 
$ p_m = \tr \left( X^m \right) $; then
the generalized cut-and-join formula  is written very compactly and beautifully:
\be\label{MM3}
{\cal W}^{\Delta}\cdot s_\mu(X)=\varphi_\mu(\Delta)s_\mu(X)
\ee
where
\be
{\cal W}^{\Delta}=\frac{1}{z_\Delta}\tr\left( D^{\Delta_1} \right)\cdots \tr\left( D^{\Delta_\ell}\right),
\ee
and the factor $z_\Delta$ is given by (\ref{z-lambda}), $D$ is
\be
D_{a,b}=\sum_{c=1}^N X_{a,c}\frac{\partial}{\partial X_{b,c}}
\ee
As G.I. Olshansky pointed out to us, this type of formula appeared in the works of Perelomov and Popov
\cite {PerelomovPopov1,PerelomovPopov2,PerelomovPopov3}
and describe the actions of the Casimir operators in the representaion $ \lambda $; see
also \cite{Zhelobenko}, Section 9.

We propose a generalization of this relation, which in our case is constructed
using a child's drawing of a constellation (dessin d'enfants, or a map in terminology \cite{ZL}.
In fact, we are considering a modification in which the vertices are replaced by small disks---"stars"). This topic will be
be studied in more detail in the next article. 
Here we restrict ourselves only to a reference to important  beautiful works
\cite{Olshanski-19,Olshanski-199,Okounkov-1,Okounkov-19,Okounkov-199}.

(i)
One can interpret the Gaussian integral as the integral of $ n $ -component two-dimensional charged bosonic fields
$Z_i$ and $Z^\dag_i$:
$$
\int (Z^\dag_i)_{a,b}(Z_j)_{b',a'}
d\Omega =
<(Z^\dag_i)_{a,b}(Z_j)_{b',a'}>=\frac 1N \delta_{a,a'}\delta_{b,b'}\delta_{i,j},
$$
for $i,j=1,\dots,n,\quad a,b=1,\dots, N$.

The Fock space of these fields is all possible polynomials from the matrix elements of the matrices
$Z_1,\dots,Z_n$.

The operators $ (Z_i)_{ab} $ can be considered creation operators, and the operators
\be\label{correspondence}
({Z}^\dag_i)\,\to\,\frac 1N \partial_i,\quad   (\partial_i)_{a,b}=\frac 1N \frac{\partial}{\partial Z_{b,a}}
\ee
ellimination operators that act in this space.
The integrands in our integrals should be considered anti-ordered, that is, all ellimination operators (all derivatives)
considered to be moved to the left, while the matrix structure is considered to be preserved. We will denote
this is anti-ordering of some $ A $ by the symbol $ :: A :: $, where $ A $ is a polynomial of matrix elements
of the matrices

From this point of view, on different sides of the ribbon of $\tig$ with the number $i$ we place the canonically 
conjugated coordinates $Z_i$ and momenta $\partial_{Z_i}$.

We recall that all partions throught the paper have the same weight $d$.

(ii)
Then, for example, the relation we get  
\begin{equation}\label{Operator-tau}
:: \tau_g(M_1,\dots,M_{\f})::
=\sum_{d=0}^\infty
\sum_{{\Delta}^1,\dots , {\Delta}^{\V}\atop |\Delta^1|=\cdots = |\Delta^{\V}|=d} 
H_{\Sigma'}(g|{\Delta}^1,\dots , {\Delta}^{\V})
{C}({\Delta}^1,\dots , {\Delta}^{\V} ),
\end{equation}

We give another relation:
\be
\left(::\prod_{i=1}^\f s_{\lambda^i}(M_i)::\right)\cdot 1 = 
\delta_{\lambda^1,\dots,\lambda^\f} \left(\frac{{\rm dim}\,\lambda}{d!}  \right)^{-n} \prod_{i=1}^\V s_\lambda(W_i^*)
\ee
where $\lambda^1,\dots,\lambda^\f=\lambda$  is a set of Young diagrams, and where $\delta_{\lambda^1,\dots,\lambda^\f}$ 
is equal 1, if $\lambda^1=\dots =\lambda^\f=\lambda$ and is equal to 0 otherwise. 

It looks like a simple rewrite, but can be helpfully used. Let us derive a beautiful formula 
(Theorem 5.1 in \cite{MM3}), namely (\ref{MM3}) (and see also 
articles \cite{Olshanski-19,Olshanski-199,Okounkov-1,Okounkov-19,Okounkov-199,Okounkov-1996}).

In order to do this we should use the freedom to choose the source matrices:

(iii)
For a partition $\Delta=(\Delta_1,\dots,\Delta_\ell)$ and a face monodromy $M_i$ and a star monodromy $W^*_i$,
let us introduce notations
\be
{\cal M}_i^{\Delta^i}=\tr \left((M_i)^{\Delta^i_1}\right)\cdots \tr \left((M_i)^{\Delta^i_\ell}\right)
\ee
\be
{\cal C}_i^{\Delta^i}=\tr \left((W_i^*)^{\Delta^i_1}\right)\cdots \tr \left((W^*_i)^{\Delta^i_\ell}\right)
\ee

We can write
\be\label{MM1-all'''matrices}
N^{nd} \int \left[ \prod_{i=1}^\f \frac{{\cal M}_i^{\tilde{\Delta}^i}}{z_{\tilde{\Delta}^i}}
 \right]d\Omega
\ee
\be 
 =\sum_{\Delta^1,\dots,\Delta^\V}
 H_\Sigma(\tilde{\Delta}^1,\dots,\tilde{\Delta}^\f,\Delta^1,\dots,\Delta^\V)
\left[ \prod_{i=1}^\V {\cal C}_i^{\Delta^i} \right]
\ee

Let us write the most general generating function for Hurwitz numbers which was obtained in~\cite{NO2020}:
\be
\label{full-generating function}
N^{nd}\int \left[ \prod_{i=1}^{\f_1} \frac{{\cal M}_i^{\tilde{\Delta}^i}}{z_{\tilde{\Delta}^i}}
 \right] 
 \left(
\prod_{i=\f_1+1}^{\f_2}{\mathfrak{M}}(M_i) \right)
\left(\prod_{i=\f_2+2,\f_2+4,\dots}^{\f}
{\mathfrak{H}}(M_{i-1},M_{i})\right)
  d\Omega
\ee
\be \label{full-rhs}
=
\sum_{{\Delta}^1,\dots , {\Delta}^{\V}\atop |\Delta^1|=\cdots = |\Delta^{\V}|=d} 
H_{\tilde{\tilde\Sigma}}(\tilde{\Delta}^1,\dots , \tilde{\Delta}^{\f_1},{\Delta}^1,\dots , {\Delta}^{\V})
{C}({\Delta}^1,\dots , {\Delta}^{\V} ),
\ee
where
\begin{myequation}
\label{Hurwitz-general}
 H_{\tilde{\tilde\Sigma}}(\tilde{\Delta}^1,\dots , \tilde{\Delta}^{\f_1},{\Delta}^1,\dots , {\Delta}^{\V})=
 \sum_{\lambda\in {\cal P}} 
 \left( \frac{{\rm dim}\,\mu}{d!} \right)^{\f-n+V-2h-m}
 \varphi_\mu(\tilde{\Delta}^1)\cdots \varphi_\mu(\tilde{\Delta}^{\f_1})
 \varphi_\mu({\Delta}^1)\cdots \varphi_\mu({\Delta}^{\V})
\end{myequation}
where $h=\frac12 (\f -\f_2)$ is the number of handles and $m=\f_2-\f_1$ is the number of
Moebius stripes glued to $\Sigma$ where the graph $\tig$  (modified dessin d'enfants) was drawn.
The Euler characteristic of ${\tilde{\tilde\Sigma}}$ is $\f-n+V-2h-m=\f_1-n+\V$. Hurwitz number
(\ref{Hurwitz-general}) counts the coverings of ${\tilde{\tilde\Sigma}}$ with branching profiles
$\tilde{\Delta}^1,\dots , \tilde{\Delta}^{\f_1},{\Delta}^1,\dots , {\Delta}^{\V}$.

Let us multiply the both sides of (\ref{full-generating function}) by
$$
 \prod_{i=k+1}^{\f_1}\frac{{\rm dim}\,\mu^i}{d!}\varphi_{\mu^i}(\tilde{\Delta}^i) z_{\tilde{\Delta}^i}
$$
(where $k\le\f$), and then sum the both sides (\ref{full-generating function}) and (\ref{full-rhs}) over 
$\tilde{\Delta}^i,\,i=k+1,\dots,\f_1$, taking into~account 
\be
s_\mu(X)=\frac{{\rm dim}\,\mu}{d!}\sum_\Delta\varphi_{\mu}(\Delta){\cal X}^\Delta ,
\ee
when evaluating  (\ref{full-generating function}),
where
\be
{\cal X}^{\Delta}=\tr \left((X)^{\Delta_1}\right)\cdots \tr \left((X)^{\Delta_\ell}\right)
\ee
and the orthogonality relation (\ref{orth2}) when evaluating  (\ref{full-rhs}). We obtain
\begin{myequation}
\label{MM1-all-klein-handle''}
\int 
N^{nd} \left[ \prod_{i=1}^k \frac{{\cal M}_i^{\tilde{\Delta}^i}}{z_{\tilde{\Delta}^i}}
 \right] 
 \left(
\prod_{i=k+\mc+1}^{k+\mc+m}{\mathfrak{M}}(M_i) \right)
\left(\prod_{i=k+\mc+m+2,k+\mc+m+4,\dots}^{\f}
{\mathfrak{H}}(M_{i-1},M_{i})\right)
\prod_{i=k+1}^{k+\mc} s_{\mu^i}(M_i)
  d\Omega
\end{myequation}
\be
=
\delta_{\mu^{k+1},\dots,\mu^{\f-k} }
 \left( \frac{{\rm dim}\,\mu}{d!} \right)^{k-n-2h-m}
 \varphi_\mu(\tilde{\Delta}^1)\cdots \varphi_\mu(\tilde{\Delta}^k)
\left[ \prod_{i=1}^\V s_\lambda(W_i^*)) \right]
\ee
where $2h=\f-\f_1-m$ $\mc$

\br\label{chess-desk}
 Suppose that the edges of the graph $\Gamma$ can be painted like a chessboard in black and white faces so that 
 the face of one color borders only the faces of a different color. Then the matrices from the set $\{ Z^\dag\}$ 
 (i.e.,~differential operators) can be assigned to the sides of the edges of white faces, that is, the matrices 
 from the set $\{Z\}$ to the sides of black faces. In this case, the monodromies of the white faces will be those 
 differential operators which will act on the monodromy of black faces.
 
\er

The most natural and simple case is the following "polarization": Suppose that $M_i,\,i=k+1,\dots,\f_1$ 
are black faces and the rest part of the face monodromies $M_i$ are while faces (see Remark \ref{chess-desk}).

(I) Let $m=h=0$.
Take as a graph $\tig$ a child's drawing - sunflower with $n$ white petals drawn on the background of black night sky.
See (b) in the Figure \ref{fig:figure4} for $\Gamma$ with 3 petals as an example. 
There is 1 vertex of $\Gamma$
which inflated and we get a small disk as the center of sunflower. We have $n+1$ faces of 
$\Gamma$: $n$ petals, and the big 
and a big face, embracing all the petals and containing infinity.
Then we place all "momentums" inside the petals:
$$
M_i=C_{-i}Z_i^\dag,\,i=1,\dots,n .
$$

Then all "coordinates" (the collections of $\{(Z_i)_a,b\}$) are placed on 
the other side 
of the ribbons, they are places along the boundary of the big embracing black face: 
$$
M_{n+1}=Z_1C_1\cdots Z_nC_n
$$

See Figure \ref{fig:figure4}b as an example.

Let remove the sign tilde above Young diagrams,
then, 
\be\label{petal-integral}
\int N^{nd} 
\left[ \prod_{i=1}^n 
\frac{{\cal M}_i^{\Delta^i}} {z_{\Delta^i}}
 \right]  
 s_{\mu}\left( Z_1C_1\cdots Z_nC_n \right)  d\Omega= 
 \varphi_\mu(\Delta^1)\cdots \varphi_\mu(\Delta^n) s_\mu(C_{-1}C_1\cdots C_{-n}C_n)
\ee

It is equivalent to
\be\label{vector-fields-action}
{\cal W}_{C_{-1}}^{\Delta^1}\cdots {\cal W}_{C_{-n}}^{\Delta^n}  \cdot s_\mu(Z_1C_1\cdots Z_nC_n)=
\varphi_\mu(\Delta^1)\cdots \varphi_\mu(\Delta^n) s_\mu(C_{-1}(Z_1)C_1\cdots C_{-n}(Z_n)C_n),
\ee
where each $ N \times N $  matrix $ C_{-i} $  can now depend, for example, polynomially on
 $Z_i$, $i=1,\dots, n$ and~where
\be
{\cal W}_{C_{-i}}^{\Delta^i}=
:\tr\left( ( C_{-i}(Z_i)\partial_i  )^{\Delta^i_1} \right)\cdots 
\tr\left( ( C_{-i}(Z_i)\partial_i  )^{\Delta^i_\ell} \right):\,,
\ee
where each $C_{-i}\partial_i $ is a matrix whose entries are differential operators, more precisely, are
the following vector fields:
\be
(C_{-i}\partial_i)_{a,b}:=\sum_{c=1}^N (C_{-i})_{a,c}\frac{\partial}{\partial (Z_i)_{b,c}}
\ee

The normal ordering indicated by two dots is the same here as in \cite{MM3}---that is,
while maintaining the matrix structure, the derivative operators do not act on
$ C_{-i} = C_{-i}(Z_i) $. Note that the normal ordering procedure is necessary in order  
the Equation (\ref {vector-fields-action}) was equivalent to the equality (\ref{petal-integral})!

The ordering is the same as in \cite{MM3}: keeping the matrix structure the derivatives do not 
act on $C_{-i}=C_{-i}(Z_i)$. 
Notice that the ordering is necessary to relate (\ref{vector-fields-action}) to (\ref{petal-integral})!

If we now take the case $ n = 1 $ (one petal), and in addition, $ C_{- i} = Z_i $, then we get the 
desired formula (\ref {MM3}).

Take another example with the same graph $ \tig $ 
with the same~monodromies. However, let $ k = 0 $.
In this case the integral (\ref{petal-integral}) can be re-written as the~relation
\be\label{handle-mob-Schur}
\hat{\mathfrak M}_1\cdots \hat{\mathfrak M}_m
\hat{\mathfrak H}_1\cdots
 \hat{\mathfrak H}_h \cdot s_\mu(Z_1C_1\cdots Z_nC_n)=
\left(\frac{{\rm dim}\,\mu}{|\mu|!}\right)^{-2h-m-n} s_\mu(C_{-1}(Z_1)C_1\cdots C_{-n}(Z_n)C_n),
\ee
where
\be
\hat{\mathfrak M}_i=:
e^{\frac 12 \sum_{j>0} \frac 1j\left(\tr \left(C_{-i}\partial_i \right)^j\right)^2
+ \sum_{j>0,{\rm odd}}\frac{1}{j}
\tr \left(\left(C_{-i}\partial_i \right)^j\right)}: \,,\quad i=1,\dots,m ,
\ee
\be
 \hat{\mathfrak H}_i= :e^{\sum_{j>0}\frac 1j \tr\left(\left(C_{1-m-2i}\partial_{m+2i-1}\right)^{j}  \right) 
\tr\left(\left(C_{-m-2i}\partial_{m+2i}\right)^{j}  \right) }:\,,\quad i=1,\dots, h.
\ee
\br
When $ C_{- i} = Z_i, \, i = 1, \dots, n $ both equalities (\ref {vector-fields-action}) and 
(\ref {handle-mob-Schur}) describe eigenvalue problems for the corresponding Hamiltonians in the two-dimensional 
bosonic theory.
Perhaps a comparison with the case analyzed by Dubrovin is appropriate. This is the case 
$ n = 1 $, $ \f = 1 $, $ {\cal W}^{(n)} $.
In this case, the operators $ {\cal W}^{(n)} \,, n = 1,2, \dots $ are the dispersionless Hamiltonians
KdV equations \cite{Dubr}.
\er

\br
The case $C_{-i}$ does not depend is also interesting in case the monodromies of the stars are degenerate 
matrices, then the whole intergal is related to the integration over rectanguler matrices. As an example
one can choose $C_{-1}C_1 \cdots C_{-n}C_n$ in (\ref{petal-integral}) as $\diag (1,1,\dots,1,0,0,\dots,0)$.
Then we get the Pochhhamer symbol in the right-hand side which allows one to related the whole integral to the
hypergeometric tau function \cite{OS-2000}. It will be discussed in a more detailed text where we
plan to relate out topic to certain topics in 
 \cite{Olshanski-19,Olshanski-199,Okounkov-1,Okounkov-19,Okounkov-199,Okounkov-1996}.
\er

Another example. $\Gamma$ has 2 vertices which are connected by 4 edges; see Figure \ref{fig:figure3}d. We 
$$
:\bpow_{\Delta^1}(\partial_1C_{-1}\partial_2C_{-2})
\bpow_{\Delta^2}(\partial_3C_{-3}\partial_4C_{-4}):  \left(
s_{\lambda}(Z_1C_1Z_4C_4)s_{\mu}(Z_2C_2Z_3C_3)\right)
$$
$$
= \delta_{\lambda,\mu} \left( \frac{{\rm dim}\,\mu}{d!} \right)^{-2}
\varphi_\mu(\Delta^1)\varphi_\mu(\Delta^2)
s_{\mu}(C_1C_{-4}C_3C_{-2})s_{\mu}(C_4C_{-1}C_2C_{-3})
$$

In particular, if one takes $C_3=C_4=C$ and 
$C_{-1}=C_{-1}(Z_2),\,C_{-2}=C_{-2}(Z_1),\,C_{-3}=C_{-3}(Z_3),\,C_{-4}=C_{-4}(Z_4) $ he gets 
$$
:\bpow_{\Delta^1}(C_{-2}(Z_1)\partial_1 C_{-1}(Z_2) \partial_2 )
\bpow_{\Delta^2}(\partial_3 C_{-3}(Z_{3})\partial_4 c_{-4}(Z_{4})):  
s_{\lambda}(Z_1 C_1 Z_{4}C )s_{\mu}( Z_{2} C_2  Z_3 C  )
$$
$$
= \delta_{\lambda,\mu}  \left( \frac{{\rm dim}\,\mu}{d!} \right)^{-2}
\varphi_\mu(\Delta^1)\varphi_\mu(\Delta^2)
s_{\mu}(C_{-2}(Z_1)C_1C_{-4}(Z_4)C)s_{\mu}(C_{-1}(Z_2)C_2C_{-3}(Z_3)C)
$$

In particular, if one takes $C_3=C_4=C$ and $C_{-1}=Z_2,\,C_{-2}=Z_1,\,C_{-3}=Z_3,\,C_{-4}=Z_4 $ (Euler~fields) 
he gets an  eigenvalue problem:
$$
:\bpow_{\Delta^1}(Z_1\partial_1 Z_2 \partial_2 )
\bpow_{\Delta^2}(\partial_3 Z_{3}\partial_4 Z_{4}):  \left(
s_{\lambda}(Z_1C_1Z_4C)s_{\mu}(Z_2C_2Z_3C)\right)
$$
$$
= \delta_{\lambda,\mu}  \left( \frac{{\rm dim}\,\mu}{d!} \right)^{-2}
\varphi_\mu(\Delta^1)\varphi_\mu(\Delta^2)
s_{\mu}(Z_1 C_1 Z_{4}C )s_{\mu}( Z_{2} C_2  Z_3 C  )
$$

In case $C_{-1}=C_{-2}=C_{-3}=C_{-4}=\mathbb{I}_N$ we get 
$$
:\bpow_{\Delta^1}(\partial_1 \partial_2 )
\bpow_{\Delta^2}(\partial_3 \partial_4 ):  \left(
s_{\lambda}(Z_1 C_1 Z_{4}C_4 )s_{\mu}( Z_{2} C_2  Z_3 C_3  )\right)
$$
$$
= \delta_{\lambda,\mu} \varphi_\mu(\Delta^1)\varphi_\mu(\Delta^2)
s_{\mu}(C_1 C_3)s_{\mu}(C_2 C_4)
$$

Now we consider another example $n=4$ with 
the graph obtained from the graph (a) in the Figures \ref{fig:figure1} and \ref{fig:figure5}  drawn on the torus
by doubling the edges: instead of each edge we draw two ones.
We have $\Gamma$ with one vertex, four edges and three faces
and obtain
$$
 :\bpow_{\Delta^1} \left(\partial_1 C_{-1}\partial_3 C_{-3}  \right)
 \bpow_{\Delta^1} \left(\partial_2 C_{-2}\partial_4 C_{-4}  \right):
 s_\lambda\left( Z_1 C_1  Z_2 C_2 Z_3 C_3   Z_4 C_4 \right)
$$ 
$$ 
 =s_\lambda\left(C_1 C_{-2} C_4 C_{-1} C_3 C_{-4} C_2 C_{-3}  \right)
$$

Take $C_{-1}=C_{-1}(Z_3),\,C_{-2}=C_{-2}(Z_2),\,C_{-3}=C_{-3}(Z_1),\,C_{-4}=C_{-4}(Z_4)$ and $C_2=C_4=C$.
As an example we obtain
$$
 :\bpow_{\Delta^1} \left(Z_1\partial_1 Z_3\partial_3   \right)
 \bpow_{\Delta^1} \left(\partial_2 Z_{2}\partial_4 Z_{4}  \right):
 s_\lambda\left( Z_1 C_1  Z_2 C Z_3 C_3   Z_4 C \right)
$$ 
$$ 
 =\left(\frac{{\rm dim}\,\lambda}{d!}   \right)^{-2}\varphi_\lambda(\Delta^1)\varphi_\lambda(\Delta^2)
 s_\lambda\left(Z_1 C_1 Z_{2} C Z_{3} C_3 Z_{4} C   \right)
$$

(iv)  Let us notice that if we take a dual graph to the sunflower graph with $n=1$ (dual to one petal~$\Gamma$, which is just
a line segment; see Figure \ref{fig:figure2}a,b), in this case we have one face and two vertices, we~get a version of the Capelli-type relation.
Then it is a task to compare explicitly such relations with beautiful results \cite{Okounkov-1,Okounkov-19,Okounkov-199,Okounkov-1996}.

(v)  There are several allusions to the existence of interesting structures related to quantum~integrability. 
First, as noted in \cite{NO2020}  by this appearance 2D Yang-Mills theory \cite{Witten}. See~also
possible connection to \cite{Gerasimov-Shatashvili}.  Then the appearance of 
the Yangians in works \cite{Olshanski-19,Olshanski-199} which, we hope, can~be related to our subject. 
And finally, the work \cite{Dubr}.

(vi) There is a direct similarity between integrals over complex matrices and integrals over unitary~matrices.
However, from our point of view direct anologues of the relations in the present paper are more involved in
the case of unitary matrices. In particular, Hurwitz numbers are replaced by a special combination of these numbers.



\begin{thebibliography}{999}

\bibitem{VKIII} Vilenkin, N.Y.; Klimyk, A.U. 
{\it Representation of Lie Groups and Special Functions. Volume 3:
Classical and Quantum Groups and Special Functions}; Kluwer
Academic Publishers: \hl{Dordrecht, The Netherlands,} 1992.

\bibitem{t'Hooft}  t'Hooft, G. A planar diagram theory for strong interactions.
\emph{Nucl. Phys.} \textbf{1974}, \emph{72}, 461--473.

\bibitem{Itzykson-Zuber}   Itzykson, C.; Zuber, J.-B. The planar approximation. II \textit{J. Math. Phys.} \textbf{1980}, \emph{21}, \hl{411.} 

\bibitem{KazakovKostovNekrasov} Kazakov, V.A.; Kostov, I.K.; Nekrasov, N.,
D-particles, Matrix Integrals and KP hierachy. 
\emph{Nucl. Phys.} \textbf{1999},~\emph{557},~413--442.
	
\bibitem{BrezinKazakov} Brezin, E.; Kazakov, V.A. Exactly solvable field theories of closed strings. \textit{Phys. Lett.} \textbf{1990}, {\it B236}, 144--150. 

\bibitem{GrossMigdal} Gross, D.J.; Migdal, A.A. A nonperturbative treatment of two-dimensional quantum  gravity. \emph{Nucl. Phys. B }\textbf{1990}, {\it 340}, 333--365.

\bibitem{Collins2016} Collins, B.; Nechita, I. Random matrix techniques in quantum information theory. \emph{J. Math. Phys.} \textbf{2016},~\emph{57},~015215.
 	
\bibitem{Preskill}  Preskill, J. Lecture Notes. Available online: \hl{http://www.theory.caltech.edu/~preskill/ph219/index.html\#lecture} \hl{(accessed on 3 AUG 2020).} 

\bibitem{Tulino} Tulino, A.M.;  Verdu, S.  Random Matrix Theory and Wireless Communications.
Foundations and Trends in Communications and Information Theory; \hl{Now Publishers Inc.: Breda, The Netherlands,} 2004; Volume 1.

\bibitem{Ak1} Akemann, G.; Ipsen, J.R.; Kieburg, M. Products of Rectangular Random Matrices: Singular Values and Progressive Scattering.  \hl{\emph{Phys. Rev. E} \textbf{2013}, \emph{88}, 052118.}

\bibitem{Ak2} Akemann, G.; Checinski, T.; Kieburg, M. Spectral correlation functions of the sum of two independent complex Wishart matrices with unequal covariances.\hl{ \emph{J. Phys. A  Math. Theor.} \textbf{2016}, \emph{49}, 315201.}

\bibitem{AkStrahov} Akemann, G.; Strahov, E. Hard edge limit of the product of two strongly coupled random matrices. \hl{\emph{Nonlinearity} \textbf{2016}, \emph{29}, 3743.
}

\bibitem{Ipsen}   Ipsen, J.R. Products of Independent Gaussian Random Matrices. \emph{arXiv} \textbf{2015}, arXiv:1510.06128; 

\bibitem{S1} Strahov, E.   Dynamical correlation functions for products of random matrices. \emph{arXiv} \textbf{2015},  arXiv:1505.02511.
     
\bibitem{S2} Strahov, E. Differential equations for singular values of products of Ginibre random matrices. \emph{arXiv} \textbf{2014},
arXiv:1403.6368.


\bibitem{Chekhov-2014} 
Ambjorn, J.; Chekhov, L.The matrix model
 for dessins d'enfants. \emph{arXiv} \textbf{2014}, arXiv:1404.4240.

\bibitem{Kazakov}
Kazakov, V.A.; Staudacher, M.; Wynter, T. Character Expansion Methods for Matrix Models of Dually Weighted Graphs. \emph{Commun. Math. Phys.} \textbf{1996}, \emph{177}, 451--468,  \hl{arXiv:hep-th/9502132.}
	

\bibitem{Kazakov2}Kazakov, V.A.; Staudacher, M.; Wynter, T. Almost Flat Planar Diagrams \emph{Commun.Math.Phys.} \textbf{1996}, \emph{179}, 235--256 
	
\bibitem{Kazakov3} Kazakov, V.A.; Staudacher, M.; Wynter, T.  Exact Solution of Discrete Two-Dimensional $R^2$ Gravity, \emph{Nucl.Phys. B471} \textbf{1996} 309-333

\bibitem{Kazakov-SolvMM}  Kazakov, V.A. Solvable Matrix Models. \emph{arXiv}\textbf{ 2000}, arXiv:hep-th/0003064.

\bibitem{Alexandrov}  Alexandrov, A. Matrix models for random partitions. \textit{Nucl. Phys. B} \textbf{2011}, {\it 851}, 620--650.

\bibitem{ChekhovAmbjorn}Ambjorn, J.; Chekhov, L.O. The matrix model
 for hypergeometric Hurwitz number.
  \emph{Theor. Math. Phys.}  \textbf{2014},~\emph{81},~1486--1498, \hl{arXiv:1409.3553.}

\bibitem{Mac}  Macdonald, I.G. {\it Symmetric Functions and Hall Polynomials}, 2nd ed.;
Clarendon Press: Oxford, UK; New~York,~NY,~USA,  1995.

\bibitem{OS-2000} Orlov, A.Y.; Scherbin, D.M. Fermionic representation for basic hypergeometric functions related to Schur~polynomials. \emph{arXiv} \textbf{2000}, arXiv:nlin/0001001

\bibitem{OS-TMP} Orlov, A.Y.; Scherbin, D.M. Hypergeometric solutions of soliton equations.
\emph{Theor. Math. Phys.} \textbf{2001},~\emph{128},~906--926.

\bibitem{O-2004-New}Orlov, A.Y. New solvable matrix integrals.
{\it  Intern.  J. Mod. Phys.} \textbf{2004}, {\it  19 } (Suppl. 02), 276--293.

\bibitem{OST-I}Orlov, A.Y.; Shiota, T.; Takasaki, K. Pfaffian structures and certain solutions to
  BKP hierarchies I. Sums over partitions. \emph{arXiv} \textbf{2012},  arXiv:math-ph/12014518. 

\bibitem{ZL}  Lando, S.K.; Zvonkin, A.K.  \emph{Graphs on Surfaces and Their Applications}; 
Encyclopaedia of Mathematical Sciences;  Zagier, D., Ed.; Springer: \hl{New York, NY, USA,}  2004; Volume 141.

\bibitem{Okounkov-Pand-2006} Okounkov, A.; Pandharipande, R. Gromov-Witten theory, Hurwitz theory and completed cycles.
\emph{Ann.~Math.} \textbf{2006}, {\it 163}, 517.


\bibitem{Dijkgraaf}  Dijkgraaf, R.  \emph{Mirror Symmetry and Elliptic Curves, The Moduli Space
of Curves};  Dijkgraaf, R.,  Faber, C.,  van~der~Geer,~G., Eds.; Progress
in Mathematics;  Birkhauser: \hl{ Basel, Switzerland, }1995; Volume 129.

\bibitem{NO2020}  Natanzon, S.M.; Orlov, A.Y.
Hurwitz numbers from matrix integrals  over Gaussian measure. accepted to  \emph{Proceedings of Symposia in Pure Mathematics}; \emph{arXiv} \textbf{2020}, arXiv:2002.00466.
 
\bibitem{NO2020F} Natanzon, S.M.; Orlov, A.Y. Hurwitz numbers from Feynman diagrams.  \emph{Theor. Math. Phys.} {\bf 2020}, {\it 204}, 1172--1199.
    
\bibitem{Chekhov-2016}  Chekhov,  L.O. The Harer-Zagier recursion for an irregular
 spectral curve. \emph{J. Geom. Phys.} \textbf{2016}, \emph{110}, 30--43, \hl{arXiv:1512.09278.}

\bibitem{O-TMP-2017} Orlov, A.Y. Hurwitz numbers and products of random matrices.
\emph{Theor. Math. Phys.} \textbf{2017}, {\it 192}, 1282--1323.

\bibitem{O-links}  Orlov, A.Y. Links between quantum chaos and counting problems.
   In \emph{Geometric Methods in Physics XXXVI, Trends in Mathematics};  Kielanowski, P., Odzijewicz, A., Previato, E., Eds.;    Birkhauser: Cham, Switzerland,  2019; pp. 355--373.

\bibitem{Amburg} Adrianov, N.M.; Amburg, N.Y.; Dremov, V.A.; Kochetkov, Y.Y.; Kreines, E.M.; Levitskaya, Y.A.; Nasretdinova,~V.F.; Shabat, G.B.   Catalog of dessins d'enfants with $\le$ 4 edges. \emph{J. Mathemat. Sci.} \textbf{2009},~\emph{158},~22--80.

\bibitem{NO2019}   Natanzon, S.M.; Orlov, A.Y. Integals of tau functions. \emph{arXiv} \textbf{2019}, arXiv:1911.02003.
    
\bibitem{Migdal} Migdal, A.A. Recursion equations in gauge field theories \emph{JETP} \textbf{1975}, \emph{42}, \hl{413.} 

\bibitem{Rusakov}  Rusakov,  B.Y.
Loop avareges and partition functions in $U(N)$ gauge theory on two-dimensional manifold.
 \emph{Mod. Phys. Lett. A} \textbf{1990}, \emph{5}, 693--703. 

\bibitem{Witten} Witten, E.  On Quantum Gauge Theories in Two Dimensions.  \emph{Com. Math. Phys.} \textbf{1991}, {\it 141}, 153--209.

\bibitem{Olshanski-19}  Olshanski, G.I. Yangians and universal enveloping algebras.  
 \emph{J. Soviet Math.} \textbf{ 1989}, \emph{47}, 2466--2473.

\bibitem{Olshanski-199} Olshanski, G.I. Representations of infinite-dimensional classical groups, 
limits of envelopingalgebras, and Yangians.  In \emph{Topics in Representation 
Theory}; Kirillov, A.A.,  Ed.; Advances in Soviet Math;   American~Mathematical~Society: Providence, RI, USA, 1991; pp. 1--66.
	
\bibitem{Okounkov-1} Okounkov, A.; Olshanski, G. {Shifted Schur functions}. \emph{Algebra i Analiz} \textbf{1998}, \emph{9}, 73--146.

\bibitem{Okounkov-19} Okounkov, A.  Shifted  Schur  functions  II.  The  binomial  formulafor  characters   of  classical  groups  and  its~applications.  In \emph{American Mathematical Society Translations}, v. 185, Kirillov's Seminar on Representation Theory; \hl{ 1998; American Mathematical Society, USA} 
 pp. 245--271.
 

\bibitem{Okounkov-199}  Okounkov, A. \textit{Quantum Immanants and Higher Capelli Identities};
Transformation Groups, v. 1 \hl{Springer, Boston} 
 1996; pp. 99--126.


\bibitem{Okounkov-1996} Okounkov, A. Young Basis, Wick Formula, and Higher Capelli Identities. \emph{Internat. Math. Res. Not.}  \textbf{1996},~\emph{17},~817--839.
	
    
\bibitem{VershikOkounkov}
Vershik, A.M.; Okounkov, A.Y. A new approach to the representation theory
of symmetric groups. II. \emph{J.~Math.~Sci.} \textbf{2005}, \emph{131},  5471--5494.

\bibitem{Gerasimov-Shatashvili} Gerasimov, A.A.; Shatashvili, S.L. 
Two-dimensional Gauge Theory and Quantum Integrable Systems. \emph{arXiv}~\textbf{2007}, arXiv:0711.1472.

\bibitem{Rumanov}  Rumanov, I.  Classical integrability for beta ensembles and general Fokker-Plank equation. \emph{J. Math. Phys.} \textbf{2015}, {\it 56}, 013508.

\bibitem{MirMorSem} Mironov, A.; Morozov, A.; Semenoff, G. Unitary Matrix Integrals in the Framework of he Generalized Kontsevich Model. \emph{Intern J. Mod. Phys. A} \textbf{1996}, \emph{11}, 5031--5080.

\bibitem{JM}  Jimbo, M.;  Miwa, T.  Solitons and Infinite
Dimensional Lie Algebras. {\it Publ. RIMS Kyoto Univ.} \textbf{1983},~{\it 19},~943--1001.

\bibitem{UT} Ueno, K.; Takasaki, K.  Toda lattice hierarchy. {\it Adv. Stud. Pure Math.} \textbf{1984}, {\it  4},  1--95.

\bibitem{Takasaki} Takasaki, K. Initial value problem for the Toda lattice hierarchy. \emph{Adv. Stud. Pure Math.} \textbf{1984}, \emph{4}, 139--163.
 
\bibitem{Takebe} Takebe, T. Representation Theoretical Meaning
of Initial Value Problem for the Toda Lattice Hierarchy I.~{
LMP}. \textbf{1991}, {\it 21}, 77--84. Takebe, T. Representation
Theoretical Meaning of Initial Value Problem for the Toda Lattice
Hierarchy II. {\it Publ. RIMS Kyoto Univ.} \textbf{1991}, {\it 27}, \hl{491--503.} 

\bibitem{ELSV} Ekedahl, T.; Lando, S.; Shapiro, M.; Vainshtein, A. On Hurwitz numbers and Hodge integrals.
\emph{Comptes~Rendus~de~l'Academie~des~Sciences-Series~I-Mathematics} \textbf{1999}, \emph{146}, 175-1180. 

\bibitem{Okounkov-2000}Okounkov,  A.   Toda equations for Hurwitz numbers. {\it Math. Res. Lett.} \textbf{2000}, {\it 7}, 447--453.

\bibitem{M1}  Mednykh, A.D.  Determination of the number of nonequivalent covering over a compact Riemann surface.
\emph{Soviet Math. Dokl.} \textbf{1978}, \emph{19}, 318--320.

\bibitem{M2}  Mednykh, A.D.; Pozdnyakova, G.G.  The number of nonequivalent covering over a compact nonorientable~surface. S\emph{ibirs. Mat.} \textbf{1986}, \emph{27}, 123-131.

\bibitem{GARETH.A.JONES} Jones, G.A. Enumeration of Homomorphisms and
Surface-Coverings.  \emph{Q. J. Math. Oxford} \textbf{1995}, \emph{46}, 485--507.

\bibitem{MM1} Mironov, A.D.; Morozov, A.Y.; Natanzon, S.M. Complect set of cut-and-join operators in the Hurwitz-Kontsevich theory. \emph{Theor. Math.Phys.} \textbf{2011}, \emph{166}, 1--22.  

\bibitem{O-1987}  Orlov, A.Y. \textit{Vertex Operator, $\bar{\partial}$-Problem, Symmetries, Variational 
Identities and Hamiltonian Formalism for 2+ 1 Integrable Systems};
Nonlinear and Turbulent Processes in Physics; World Scientific: Singapore, 1987.

\bibitem{MM3} Mironov, A.; Morozov, A.; Natanzon, S.
{Algebra of differential operators associated with Young diagramms}. \textit{J. Geom. 
Phys.} \textbf{2012}, \emph{62},  148--155.

\bibitem{PerelomovPopov1} Perelomov, A.M.;  Popov, V.S.   
Casimir operators for groups
$U(N)$ $SU(N)$. \emph{Yadernaya Fizika} \textbf{1966}, {\it 3}, 924--931.


\bibitem{PerelomovPopov2} Perelomov, A.M.;  Popov, V.S.    
Casimir operators for classical groups. \emph{Doklady AN SSSR} \textbf{1967}, {\it 174},
287--290.(In~Russian)


\bibitem{PerelomovPopov3} Perelomov, A.M.;  Popov, V.S.   
Casimir operators for semisimple Lie groups.
\emph{Izavestia AN SSSR} \textbf{1968},~\emph{32},~1368--1390.

\bibitem{Zhelobenko}  Zhelobenko, D.P. \textit{Compact Lie Groups and Their Representations}; American Mathematical Soc.: Providence,~RI,~USA,  1973.

\bibitem{Dubr}  Dubrovin, B.A. Symplectic field theory of a disk, quantum integrable 
systems, and Schur polynomials. \textit{arXiv}~\textbf{2016}, arXiv:1407.5824.







	

    
    


 
	
 







	
	
 



  


	








  

  
  
 
 

 


 


	

		


   






 














\end{thebibliography}
\end{document}